\documentclass[12pt]{article}
\usepackage[margin=1in]{geometry}
\usepackage[T1]{fontenc}
\usepackage[utf8]{inputenc}
\usepackage{lmodern,microtype}
\usepackage{amsmath,amssymb,amsthm,mathtools}
\usepackage{enumitem}
\usepackage{booktabs,tabularx}
\usepackage{graphicx}
\usepackage{setspace}
\usepackage{needspace}
\usepackage[round,authoryear]{natbib}
\usepackage{xcolor}
\usepackage{xurl}
\usepackage[colorlinks=true,linkcolor=black,citecolor=black,urlcolor=black]{hyperref}
\numberwithin{equation}{section}
\newtheorem{theorem}{Theorem}
\newtheorem{proposition}{Proposition}
\newtheorem{lemma}{Lemma}[section]
\theoremstyle{definition}
\newtheorem{assumption}{Assumption}
\BeforeBeginEnvironment{assumption}{\Needspace{3\baselineskip}}
\newtheorem{example}{Example}
\BeforeBeginEnvironment{theorem}{\Needspace{4\baselineskip}}
\BeforeBeginEnvironment{proposition}{\Needspace{4\baselineskip}}
\BeforeBeginEnvironment{lemma}{\Needspace{6\baselineskip}}
\DeclareMathOperator{\Bin}{Bin}
\newcommand{\E}{\mathbb E}
\hypersetup{pdfauthor={Keita Kuwahara},pdftitle={The Capacity Cost of Informational Screening},
  pdfkeywords={information design, screening without transfers, private beliefs, recommendation mechanisms, capacity constraints}}
\title{The Capacity Cost of Informational Screening}
\author{Keita Kuwahara\thanks{Graduate School of Economics, The University of Tokyo,
7-3-1 Hongo, Bunkyo-ku, Tokyo 113-0033, Japan.
Email: \href{mailto:kuke0303@g.ecc.u-tokyo.ac.jp}{\texttt{kuke0303@g.ecc.u-tokyo.ac.jp}}}}
\date{}
\begin{document}

\maketitle
\begin{abstract}
\noindent
An advisor knows who would benefit from using a resource but cannot assign its use. Individuals have private assessments and may ignore advice. We study how their choices can respect a fixed usage limit. In our model, whenever coordination is possible, eliciting private assessments makes some harmful use unavoidable. Policies that maximize welfare, or even minimize harm, also forgo beneficial uses that direct assignment would permit. In those environments, observing assessments instead allows advice to achieve as much beneficial use as direct assignment, without harmful use. The loss therefore reflects the incentive cost of eliciting private beliefs while leaving use voluntary.
\end{abstract}
\vspace{0.5em}

\newpage

\section{Introduction}
\label{sec:introduction}

Can expert advice coordinate voluntary use of a resource when aggregate use must stay within a fixed limit? An advisor may know who would benefit from use but lack authority to assign it or make transfers. Individuals have their own assessments and can disregard advice. We ask how the advisor can learn these assessments and give recommendations that people will follow while keeping use within the limit.

The central finding is that eliciting private beliefs consumes capacity. An advisor who could assign use directly would enable as much beneficial use as the limit permits and prevent harmful use. When every open interval of beliefs between zero and one has positive probability, any feasible advice must induce some harmful use. More surprisingly, every policy maximizing welfare or minimizing harm induces less beneficial use than direct assignment. Even when advice can preserve all the beneficial use attainable by direct assignment, a small sacrifice improves both objectives. In every environment where private-belief coordination is feasible, observing beliefs instead allows advice that agents follow to attain the same beneficial use as direct assignment, without harmful use.

Antibiotic stewardship motivates the distinction between advice and control. Unnecessary antibiotic use contributes to resistance, making infections harder to treat \citep{WHO2026}. Laboratories can influence prescribing by selectively reporting antibiotic susceptibility results \citep{CDC2020}. In prospective audit and feedback, stewardship teams review prescriptions and recommend changes for the prescribing doctor to consider \citep{Barlam2016}. These practices motivate the advisory channel studied here.\footnote{Stewardship programs may also require approval before certain antibiotics can be prescribed \citep{Barlam2016}; that authority is outside the advisory remit considered here.}

Consider a stylized advisory service asked to coordinate doctors' choices between a specified antibiotic and standard treatment under a fixed usage ceiling for a group of patients. The service has an advisory remit: it controls neither treatment access nor monetary incentives. Both treatment options remain available, so a doctor's deviation can exceed the ceiling. Each doctor privately assesses the probability that the specified antibiotic offers a positive expected patient benefit relative to standard treatment. We idealize the service's diagnostic information as identifying the sign of that benefit. The service therefore uses reported assessments to tailor advice to doctors' beliefs; the reports add nothing to its diagnosis.

Revealing the diagnosis to every doctor may violate the ceiling: whenever too many patients would benefit, all their doctors would choose the antibiotic. Coordinating use instead requires recommendations of standard treatment that some doctors will follow given their own assessments. The fixed group ceiling and decisive diagnostic information are modeling assumptions. The ceiling is an exogenous policy target; the model does not derive it from resistance. Our welfare comparison measures patient benefits and losses within that limit, without separately valuing resistance effects on future patients.\footnote{Some selectively omitted susceptibility results are available on request \citep{CDC2022}; we abstract from further diagnostic inquiry.}

The model generalizes this advisory problem. Each agent privately knows its probability of being suitable for use; the advisor observes suitability but not these beliefs. A suitable user gains, an unsuitable user loses, and nonuse yields zero. The advisor commits to a reporting and private recommendation policy. Reports reveal agents' willingness to follow different advice, without adding information about suitability. We call elicitation through report-dependent advice \emph{informational screening}. Agents may change both their report and their subsequent action. Capacity must hold with probability one in the specified equilibrium, while both actions remain available to every deviator. The limit applies separately to each decision group: low use in one group cannot offset excess use in another.

Why does elicitation create losses? More confident agents place more weight on suitable-state gains and less on unsuitable-state losses. Advice can therefore screen beliefs by offering a greater probability of use when suitable together with a greater probability of use when unsuitable, which we call \emph{unsuitable exposure}. Unsuitable exposure acts as a price for suitable-state use probability, and both consume capacity. Concentrating the screening burden on beliefs sufficiently close to certainty can reduce its incidence because those agents are rarely unsuitable. Yet preserving every place that direct assignment would give a suitable agent limits this adjustment: unsuitable users must fit into places left over by suitable agents. Theorem~\ref{thm:distortion} shows that every optimum crosses that limit and sacrifices beneficial use. The conclusion covers mechanisms that treat agents differently and smooth nonlinear objectives that value beneficial use and penalize harm.

A three-agent example makes the tradeoff concrete. With a ceiling of two uses, uniform beliefs, and equal gains and losses, welfare-maximizing advice reduces beneficial use by about $0.06\%$ and harmful use by about $34\%$ relative to the least harmful policy that preserves all beneficial use attainable by direct assignment. Within each policy, low-confidence reports receive a common suitable-state use probability without unsuitable exposure; higher reports obtain more suitable-state use only with unsuitable exposure. Optimal advice delays the start of exposure to higher reported confidence. Section~\ref{sec:example} constructs and compares these policies.

The analysis combines individual incentives with a joint capacity restriction. Theorem~\ref{thm:implementation} characterizes reporting and action incentives through a distribution of belief thresholds. It then gives necessary and sufficient restrictions for combining the resulting individual recommendations within the realized ceiling. Checking total expected use alone would not suffice: every group of belief--state pairs must fit within its own capacity bound.

The same friction limits when coordination is possible at all. Proposition~\ref{prop:feasible-costs} compares private-belief preservation of beneficial use, any private-belief coordination, and observed-belief coordination. Whenever private-belief coordination is possible, these three ranges of benefits from use are strictly nested. With uniform beliefs, a ceiling of one use precludes all private-belief coordination; a ceiling of two or more permits preservation when benefits are sufficiently small. Finally, Proposition~\ref{prop:finite-reports} shows that the number of agents restricted to finitely many possible reports must be smaller than the usage ceiling, although each agent need receive only a recommendation to use or forgo the resource.

\paragraph{Related literature.}
The closest literature studies persuasion of privately informed receivers. With binary actions, private preference types independent of the state, and linear payoffs, \citet[Theorem~1 and Corollary~1]{KMZL2017} show that a single report-independent disclosure policy can reproduce each type's expected utility and action probability under a report-dependent mechanism. \citet{GuoShmaya2019} allow private information about project quality and show that optimal disclosure need not elicit it; they interpret the endpoints of acceptance intervals as quantity and price in a screening problem. With more than two actions, screening can substantially improve disclosure \citep{CandoganStrack2023}. In our environment, every open interval of beliefs has positive probability, and the realized ceiling makes elicitation necessary even with binary actions and payoffs that do not depend on others' actions. The screening price itself consumes capacity, and minimizing its incidence requires sacrificing beneficial use.

Our advisor uses committed disclosure and correlated private recommendations \citep{KamenicaGentzkow2011,BergemannMorris2016,ArieliBabichenko2019,Taneva2019}. \citet{BergemannMorris2019} distinguish observing private information from eliciting it and show that elicitation can matter even when payoffs do not depend on other receivers' actions. \citet{CastiglioniMarchesiGatti2022} study computational gains from type-report menus in a multi-receiver model with finitely many private types and binary actions. \citet{BonattiDahlehHorel2026} study elicitation in linear--quadratic games, allowing false reports followed by disobedience. We impose these joint-deviation constraints and characterize which individual reporting and action incentives can be satisfied together within a realized ceiling.

Other screening models give the designer control over allocation \citep{Miralles2012,DoganUyanik2020,PereyraSilva2023} or available actions \citep{Vairo2025}. \citet{Dogan2026} uses sequential disclosure to initially uninformed agents. \citet{SatoShirakawa2026} elicit private beliefs about common-value goods; their supply extension constrains average allocation. \citet{GuoYan2026} combine treatment allocation with recommendations to a separate receiver. Our advisor observes suitability but cannot exclude anyone from acting.

Capacity affects information design through seller availability in \citet{RomanyukSmolin2019}, enforcement resources in \citet{HernandezNeeman2022}, and queue congestion in \citet{AnunrojwongIyerManshadi2023}. Here the ceiling links otherwise independent decisions. In antibiotic prescribing, \citet{RibersUllrich2024} combine algorithms with physicians' private information to reduce use while preserving the observed number of treated bacterial urinary tract infections, assuming compliance. We study compliance incentives, using the direct-assignment maximum of beneficial use as the preservation benchmark.

Section~\ref{sec:model} introduces the environment. Section~\ref{sec:costs} states the capacity cost of screening and compares private and observed beliefs. Section~\ref{sec:implementation} characterizes implementation and explains the cost. Section~\ref{sec:example} develops the welfare-maximizing example. Section~\ref{sec:feasibility} identifies when coordination is possible and why finite reports cannot achieve it. The appendices contain proofs and the example's exact construction.

\section{The advisory environment}
\label{sec:model}
\subsection{Beliefs, states, and payoffs}

There are $n\ge2$ agents who each decide whether to use a resource, with a usage ceiling $k\in\{1,\ldots,n-1\}$. Agent $i$ privately observes a belief $p_i\in(0,1)$ about a binary state $\theta_i\in\{0,1\}$. State one is \emph{suitable} for use; state zero is \emph{unsuitable}. Beliefs have a common distribution $F$ with no point masses (an \emph{atomless} distribution). Except where explicitly relaxed, we impose:
\begin{assumption}
\label{ass:belief-support}
$F$ has full support on $(0,1)$.
\end{assumption}
Thus every open subinterval has positive probability, including intervals arbitrarily close to certainty.
The pairs $(p_i,\theta_i)$ are independent across agents, with
\[
 p_1,\ldots,p_n\overset{\mathrm{iid}}{\sim}F,
 \qquad
 \theta_i\mid p_i\sim\operatorname{Bernoulli}(p_i).
\]
Thus beliefs are calibrated: $p_i$, also called the agent's type, is its probability of suitability. Write $\mathbf p=(p_i)_{i=1}^n$ and $\boldsymbol\theta=(\theta_i)_{i=1}^n$. The advisor observes $\boldsymbol\theta$ but not $\mathbf p$; each agent observes only its own $p_i$. The joint law of a belief--state pair is
\begin{equation}
 \rho(dp,1)=p\,dF(p),\qquad \rho(dp,0)=(1-p)\,dF(p).
 \label{eq:belief-state-law}
\end{equation}

Agent $i$ chooses an action $a_i\in\{0,1\}$: use of the resource ($a_i=1$) or nonuse ($a_i=0$). Write $\mathbf a=(a_i)_{i=1}^n$ for the action vector. Its payoff relative to nonuse is
\begin{equation}
 u_i(a_i,\theta_i)=a_i\big[u_+\theta_i-u_-(1-\theta_i)\big],
 \qquad u_+>0,\quad u_->0.
 \label{eq:payoffs}
\end{equation}
The common gains and losses are fixed parameters and do not depend on others' actions. All primitive distributions and payoff parameters are common knowledge. Suitability describes the sign of the use payoff; it is distinct from an agent's belief about that sign.

In the prescribing example, agent $i$ is a doctor deciding between the specified antibiotic ($a_i=1$) and standard treatment ($a_i=0$) for one patient. Nonuse here means standard treatment, not withholding care. State $\theta_i=1$ means that the antibiotic has a positive expected benefit relative to standard treatment; $\theta_i=0$ means that its expected benefit is negative. These states concern the expected value of the treatment choice, not whether the patient ultimately recovers. The doctor's initial assessment is $p_i$, and the advisor's additional diagnostic information identifies $\theta_i$. We assume that conditional gains and losses are common across doctors, so their private information concerns the probability of benefit.

\Needspace{8\baselineskip}
\subsection{Communication and voluntary use}

The advisor commits to a reporting and private recommendation policy before beliefs and states are realized. Agents submit reports simultaneously, receive private recommendations, and then choose actions simultaneously. The advisor observes the states and all reports; each agent observes only its own type, report, recommendation, and private randomization. There is no further diagnostic information before actions. Both actions remain available after every report and recommendation: the advisor cannot make transfers, compel actions, or deny access.

We can restrict attention to \emph{direct recommendation mechanisms}. Reports and types range over $(0,1)$. Writing $\mathbf r=(r_i)_{i=1}^n$ for the report vector, such a mechanism is a measurable probability kernel $\pi(\mathbf m\mid\mathbf r,\boldsymbol\theta)$ on recommendation vectors $\mathbf m=(m_i)_{i=1}^n\in\{0,1\}^n$, defined at every report--state profile $(\mathbf r,\boldsymbol\theta)\in(0,1)^n\times\{0,1\}^n$. The advisor privately sends recommendation $m_i$ to agent $i$. Under the \emph{truthful-and-obedient} strategy profile $\sigma^{\mathrm{TO}}$, each agent reports $r_i=p_i$ and takes $a_i=m_i$. Write $\Pr_\pi$ and $\E_\pi$ for the resulting joint law and expectation, whether or not this profile is an equilibrium.

Truthful obedience must be a Bayesian Nash equilibrium against \emph{joint reporting--action deviations}: an agent may change its report and then choose either action as a function of its type, report, and recommendation. In particular, it may ignore or reverse advice following a false report. Optimality is evaluated under $F$, so it is required for almost every true type; every report in $(0,1)$ remains available as a deviation. Section~\ref{sec:implementation} gives the exact incentive inequalities conditional on the agent's type.

An \emph{outcome} is the induced joint law of $(\mathbf p,\boldsymbol\theta,\mathbf a)$. A direct mechanism is \emph{feasible} if truthful obedience is an equilibrium and
\begin{equation}
 \sum_{i=1}^n a_i\le k\quad\Pr_\pi\text{-almost surely}.
 \label{eq:capacity}
\end{equation}
We call \eqref{eq:capacity} the \emph{hard capacity constraint}: it is a realized usage ceiling, not a restriction on the actions agents can choose. Deviations remain available, with the payoffs in \eqref{eq:payoffs}, even when they cause total use to exceed $k$. We require implementation in a specified equilibrium, not in every equilibrium. Coordination is possible when a feasible mechanism exists.

This direct formulation loses no outcomes. Lemma~\ref{lem:a-revelation} shows that every feasible equilibrium outcome of a mechanism with arbitrary standard Borel report and message spaces\footnote{Standard Borel spaces are measurable spaces isomorphic to Borel subsets of complete separable metric spaces; they include finite and countable sets and Euclidean report or message spaces.} has a direct implementation with the same joint outcome law. Appendix~\ref{app:general-communication} specifies the larger class and proves the reduction, including joint deviations. Thus our optimization and impossibility results apply to that entire class, and references to a feasible general mechanism include its specified equilibrium.

\subsection{The advisor's objective and the direct-assignment benchmark}

For a feasible outcome, define welfare $W$, beneficial use $S$, and harmful use $R$ by
\begin{equation}
 \begin{aligned}
 W&=u_+S-u_-R,\\
 S&=\E\sum_i\theta_i a_i,\qquad
 R=\E\sum_i(1-\theta_i)a_i.
 \end{aligned}
 \label{eq:participation-welfare}
\end{equation}
We call $S$ and $R$ \emph{beneficial use} and \emph{harmful use}: they are the expected numbers of uses in the suitable and unsuitable states, respectively. The advisor maximizes $\Phi(S,R)$, where $\Phi$ is continuously differentiable on $[0,k]^2$, with
\begin{equation}
 \Phi_S\ge0,\qquad \Phi_R<0.
 \label{eq:advisor-marginals}
\end{equation}
The subscripts on $\Phi$ denote partial derivatives. There is no additional constraint on $S$. In particular, minimizing $R$ does not mean that the advisor can shut down use: recommending nonuse to everyone would be rejected by sufficiently confident agents. The objective changes neither agents' payoffs nor the feasible set. Examples include welfare $\Phi(S,R)=W=u_+S-u_-R$ and harm minimization $\Phi(S,R)=-R$, but linearity is not assumed. In the diagnostic interpretation, $R$ counts uses with negative expected patient benefit. Welfare sums individual payoffs within the prescribed usage ceiling, without a separate resistance externality.

Let $\Bin(N,\eta)$ denote a binomial random variable with an integer $N\ge0$ trials and success probability $\eta\in[0,1]$. Define
\begin{equation}
 \begin{gathered}
 \mu=\int_0^1p\,dF(p),\qquad
 H(\eta)=\E\min\{k,\Bin(n,\eta)\},\\
 S^{\mathrm{FB}}=H(\mu),\qquad g=\frac{S^{\mathrm{FB}}}{n\mu}.
 \end{gathered}
 \label{eq:first-best}
\end{equation}
Direct assignment attains the \emph{first-best beneficial use} $S^{\mathrm{FB}}$ by selecting only suitable agents, up to capacity. It has $R=0$ and welfare $u_+S^{\mathrm{FB}}$. Uniform selection among suitable agents gives each use probability $g\in(0,1)$ conditional on suitability. Every feasible outcome satisfies
\begin{equation}
 \sum_i\theta_i a_i\le\min\Bigl\{k,\sum_i\theta_i\Bigr\}
 \quad\text{almost surely},
 \label{eq:first-best-bound}
\end{equation}
so $S\le S^{\mathrm{FB}}$, with equality precisely when \eqref{eq:first-best-bound} holds with equality almost surely.

An expected capacity constraint would be weaker. If only $\E\sum_i a_i\le k$ were required and $n\mu\le k$, privately revealing each $\theta_i$ would implement $S=n\mu$, $R=0$ without reports. This rule exceeds the realized ceiling with positive probability because $0<\mu<1$ and $n>k$.

\section{The costs of informational screening}
\label{sec:costs}

We say that an outcome \emph{preserves beneficial use} if $S=S^{\mathrm{FB}}$; it may still induce harmful use. For comparison, suppose an \emph{observed-belief} advisor privately observes both $\mathbf p$ and $\boldsymbol\theta$. This advisor need not elicit reports, but agents still choose their actions and capacity must still hold. The next theorem isolates the cost of eliciting beliefs while keeping these other restrictions fixed.

\begin{theorem}[The capacity cost of screening]
\label{thm:distortion}
Suppose a feasible mechanism exists when beliefs are private.
\begin{enumerate}[label=(\roman*),leftmargin=*]
\item The minimum harmful use $\min R$ is attained and strictly positive.
\item The objective $\Phi(S,R)$ attains a maximum, and every maximizer satisfies $S<S^{\mathrm{FB}}$.
\item If the advisor privately observes both beliefs $\mathbf p$ and states $\boldsymbol\theta$, obedient advice attains $(S,R)=(S^{\mathrm{FB}},0)$ under the same usage ceiling.
\end{enumerate}
\end{theorem}

The theorem applies to all feasible mechanisms, including asymmetric ones. Attainment follows from compactness of the attainable set of $(S,R)$ (Lemma~\ref{lem:b-attainment}); Appendix~\ref{app:capacity-cost} proves the theorem.

\paragraph{Why some harmful use is unavoidable.}
If harmful use were zero, unsuitable exposure would vanish at almost every truthful report. Every type would rank these reports solely by their suitable-state use probability, so truthfulness would require the same probability at almost every truthful report. Because types can approach certainty, this common probability would have to be one to make nonuse obedient. All suitable agents would consequently use the resource, exceeding capacity whenever more than $k$ are suitable. Compactness strengthens this impossibility to a positive minimum: no sequence of feasible policies can make harm vanish.

\paragraph{Why minimizing harm also sacrifices beneficial use.}
Part (ii) applies both to welfare and to harm minimization. Unsuitable exposure is needed to elicit beliefs, but its incidence depends on who bears it. Concentrating the screening burden on beliefs sufficiently close to certainty can reduce its expected incidence. Preserving all beneficial use limits that shift because an unsuitable user must fit into a place left by suitable agents. Section~\ref{sec:implementation} makes this conflict precise.

The result has a stronger local implication whenever preservation is feasible. Given a feasible pair $(S^{\mathrm{FB}},R)$ and any $M>0$, there are feasible pairs $(S',R')$ arbitrarily close to it such that
\begin{equation}
 S'<S^{\mathrm{FB}},\qquad R'<R,\qquad
 R-R'>M(S^{\mathrm{FB}}-S').
 \label{eq:local-improvement}
\end{equation}
Thus the reduction in harm per suitable user sacrificed has no finite upper bound. This compares attainable pairs $(S,R)$ without assuming that their frontier is differentiable. Choosing $M>u_+/u_-$ improves welfare as well as reducing harm. Appendix~\ref{app:capacity-cost} proves \eqref{eq:local-improvement} by improving an auxiliary objective and mixing the resulting policy with the original one. Section~\ref{sec:feasibility} identifies when preservation is feasible.

\paragraph{Why observing beliefs removes the cost.}
Every private-belief outcome remains attainable when the advisor observes beliefs: it can substitute observations for truthful reports. Starting from that policy, discard unsuitable users and add suitable nonusers until no more fit, retaining every initially selected suitable agent. Send only the final recommendations. A recommendation to use the resource now reveals suitability. Nonuse becomes easier to recommend because it is less likely to exclude someone suitable and can avoid more harmful use.

Formally, let $x_i(p)$ and $y_i(p)$ denote the original use probabilities conditional on belief $p$ and the suitable and unsuitable states, respectively. Obedience to nonuse requires
\begin{equation}
 u_+p[1-x_i(p)]\le u_-(1-p)[1-y_i(p)]
 \quad F\text{-almost everywhere}.
 \label{eq:observed-obedience}
\end{equation}
The improvement weakly raises $x_i(p)$ and sets $y_i(p)$ to zero, preserving this inequality. It need not preserve truthfulness when $p$ must be reported: eliminating exposure also eliminates the price of a higher suitable-state use probability.

The observed-belief policy attains $(S,R)=(S^{\mathrm{FB}},0)$ and strictly improves the advisor's value. Indeed, writing $R_{\min}>0$ for the minimum in part (i), every private-belief outcome satisfies
\[
 \Phi(S,R)\le\Phi(S^{\mathrm{FB}},R_{\min})
 <\Phi(S^{\mathrm{FB}},0).
\]
The comparison matches the aggregate benchmark, rather than a specified direct-assignment lottery over users. It removes only elicitation; voluntary use and capacity remain unchanged. Lemma~\ref{lem:observed-beliefs} formalizes the construction and shows that any feasible observed-belief policy can be made first best, even without full support.

\paragraph{The role of extreme beliefs.}
The availability of beliefs arbitrarily close to certainty matters for the private-belief cost. Write $\operatorname{supp}F$ for the support of the belief distribution. When Assumption~\ref{ass:belief-support} is dropped, with all other assumptions maintained, a feasible private-belief mechanism can attain $R=0$ if and only if
\begin{equation}
 \sup\operatorname{supp}F\le\frac{u_-}{u_-+u_+(1-g)}.
 \label{eq:zero-unsuitable}
\end{equation}
Under this condition, uniformly selecting suitable agents up to capacity, without reports, implements $(S^{\mathrm{FB}},0)$. The condition is exactly obedience to nonuse: $u_+p(1-g)\le u_-(1-p)$. Thus the positive minimum in Theorem~\ref{thm:distortion} depends on the presence of agents sufficiently confident to reject this advice. Lemma~\ref{lem:a-support} proves necessity even for report-dependent mechanisms.

\section{Implementing advice under capacity}
\label{sec:implementation}

\subsection{Belief reports as a screening menu}

Two restrictions determine which advice can be implemented. Each agent must prefer its assigned option to every option obtainable by misreporting and then disobeying. In addition, the options offered to different agents must fit together within capacity. We first describe the individual options and then the joint restriction.

With private beliefs, an agent chooses its state-contingent recommendation probabilities through its report. The subscript $-i$ denotes all agents other than $i$, whose belief--state pairs have product law $\rho^{n-1}$. For a direct kernel $\pi$, an agent $i$, and a report $r\in(0,1)$, define the probabilities conditional on own report and state, averaged over the other agents (the \emph{interim probabilities}), by
\begin{equation}
 \begin{aligned}
 x_i(r)&=\E_{\rho^{n-1}}\!\left[
   \pi(m_i=1\mid(r,\mathbf p_{-i}),(1,\boldsymbol\theta_{-i}))\right],\\
 y_i(r)&=\E_{\rho^{n-1}}\!\left[
   \pi(m_i=1\mid(r,\mathbf p_{-i}),(0,\boldsymbol\theta_{-i}))\right].
 \end{aligned}
 \label{eq:interimrules}
\end{equation}
The inserted report and state occupy coordinate $i$; the expectation averages over truthful opponents with law $\rho^{n-1}$. Thus $(x_i(r),y_i(r))$ is the option available to any true type $p\in(0,1)$ reporting $r\in(0,1)$. We call $y_i(r)$ \emph{unsuitable exposure}. When recommendations are obeyed, these are the corresponding state-contingent use probabilities. The pairs available through an agent's reports form its \emph{menu}. Across agents, these functions form the \emph{interim reduced form}, symmetric when these functions coincide across agents. A kernel is symmetric if simultaneous relabeling of inputs and recommendations leaves it unchanged. 

For each agent, suppress the index and define, for $p,r\in(0,1)$,
\begin{equation}
 U(p;r)=u_+p x(r)-u_-(1-p)y(r),\qquad
 U(p):=U(p;p),\qquad
 U^{1}(p)=u_+p-u_-(1-p).
 \label{eq:interim-utilities}
\end{equation}
The four deterministic action rules after report $r$---always forgo use, always use the resource, obey, and reverse---yield $0,U^{1}(p),U(p;r)$, and $U^{1}(p)-U(p;r)$. We call an interim pair \emph{incentive compatible} if truthful reporting and obedience weakly dominate every joint reporting--action deviation for every true type. Equivalently, for every $p,r\in(0,1)$,
\begin{equation}
 U(p)\ge\max\{0,U^{1}(p),U(p;r),U^{1}(p)-U(p;r)\}.
 \label{eq:jointIC}
\end{equation}
For truthful obedience to be an equilibrium, these inequalities need hold only for $F$-almost every true type, simultaneously for every report. The comparison with zero is individual rationality relative to nonuse. The first two comparisons impose obedience after the truthful report; the others allow false reports followed by obedience or reversal. Randomized deviations are mixtures. The exceptional set of true types is common to all reports; reports cannot be excluded merely because they have probability zero. An incentive-compatible reduced form is \emph{jointly implementable under hard capacity} if one feasible direct kernel has its interims at every report.

Every feasible outcome has a direct implementation with right-continuous, incentive-compatible interims, called \emph{canonical rules}. Changing recommendations only at an $F$-null set of reports makes these inequalities hold at every type without changing the outcome. To symmetrize, randomly relabel agents before applying the policy, then undo the relabeling without revealing it. Averaging over this random relabeling gives symmetric interims with the same $(S,R)$.\footnote{Lemmas~\ref{lem:a-revelation}--\ref{lem:a-symmetry} establish the direct reduction following \citet{Myerson1982}, canonical completion, and symmetrization, accounting for joint reporting and action deviations.}

The basic screening tradeoff can be seen before solving these inequalities. Compare two options, one offering an additional amount $\Delta x>0$ of suitable-state use probability and an additional amount $\Delta y>0$ of unsuitable exposure. A type $p$ is indifferent between them when
\[
 u_+p\,\Delta x=u_-(1-p)\,\Delta y.
\]
Higher types prefer the option that raises both probabilities; lower types prefer the one that lowers both. Thus unsuitable exposure plays the role of a price for increased suitable-state use, even though it is implemented through information rather than a payment.

Write $v(t)=t/(1-t)$ for belief odds, $t\in[0,1)$. Theorem~\ref{thm:implementation}(i) below makes the tradeoff exact: incentive-compatible canonical rules are nondecreasing and satisfy $u_-\,dy=u_+v\,dx$, with $x(1-)=1$ and $y(0+)=0$. Here $dx$ and $dy$ are the Stieltjes measures recording increases in $x$ and $y$, and arguments $t-$ and $t+$ denote one-sided limits from below and above. A type arbitrarily close to certainty must have suitable-state use probability approaching one, since it can always ignore advice and use the resource. At the opposite end, a type arbitrarily close to zero must have unsuitable exposure approaching zero, since it can always forgo use.

We record increases in $x$ by the probability measure $Q=x(0+)\delta_0+dx$ on $[0,1)$, where $\delta_z$ denotes the probability measure placing unit mass at $z$. Its mass at zero is baseline suitable-state use probability; its mass over any interval records the increase in suitable-state use probability offered across that interval. An increase $dQ$ at threshold $\tau$ is accompanied by unsuitable exposure $(u_+/u_-)v(\tau)dQ$ at all higher reports. This gives the \emph{threshold representation} in the next theorem. We call $\int v\,dQ$ the \emph{odds moment}; the \emph{odds bound} in the theorem limits this moment so that the unsuitable-state probability never exceeds one. The representation describes interim probabilities; it need not be a lottery over separately feasible experiments.

\subsection{Fitting individual recommendations within capacity}

Satisfying an agent's incentives does not ensure that its recommendations can be combined with everyone else's. For any group defined by beliefs and states, expected use from that group cannot exceed the expected number of its members that fit within $k$ places. Define the state-specific upper-tail masses
\begin{equation}
 T_1(t)=\int_t^1p\,dF(p),\qquad
 T_0(t)=\int_t^1(1-p)\,dF(p),\qquad 0\le t\le1.
 \label{eq:tail-masses}
\end{equation}
They are the probabilities that an agent's belief exceeds $t$ and its state is suitable or unsuitable, respectively. In particular, $T_1(0)=\mu$ and $T_0(0)=1-\mu$.

For cutoffs $s_1,s_0\in[0,1]$, consider the group
\begin{equation}
 \mathcal B_{s_1,s_0}=\{(p,1):p>s_1\}\cup\{(p,0):p>s_0\},
 \qquad \rho(\mathcal B_{s_1,s_0})=T_1(s_1)+T_0(s_0).
 \label{eq:capacity-upper-sets}
\end{equation}
The number of agents in this group is binomial. The function $H$ from \eqref{eq:first-best} gives the expected number of them that capacity can accommodate. Expected use from the group must therefore be no larger than $H(T_1(s_1)+T_0(s_0))$. The \emph{subset capacity inequalities} in part (ii) below impose this restriction on every cutoff pair. The theorem shows that these restrictions suffice once individual incentives hold.

\begin{theorem}[Incentives and joint implementation]
\label{thm:implementation}
Let $x(r)$ and $y(r)$ be the interim probabilities of a recommendation to use, conditional on report $r$ and on the suitable and unsuitable states, respectively, averaged over truthful opponents. Suppose $(x,y):(0,1)\to[0,1]^2$ is right-continuous.
\begin{enumerate}[label=(\roman*),leftmargin=*]
\item Truthful reporting and obedience weakly dominate every joint reporting--action deviation for every type if and only if, for some probability measure $Q$ on $[0,1)$,
\begin{equation}
 x(p)=x_Q(p):=Q([0,p]),\qquad
 y(p)=y_Q(p):=\frac{u_+}{u_-}\int_{[0,p]}\frac{\tau}{1-\tau}\,Q(d\tau),
 \label{eq:Qrepresentation}
\end{equation}
with
\begin{equation}
 \int_{[0,1)}\frac{\tau}{1-\tau}\,Q(d\tau)\le\frac{u_-}{u_+}.
 \label{eq:incentive-moment}
\end{equation}
\item Such a pair has a symmetric feasible direct implementation if and only if, for all $s_1,s_0\in[0,1]$,
\begin{equation}
 \begin{aligned}
 &n\int_{s_1}^1p x(p)\,dF(p)+n\int_{s_0}^1(1-p)y(p)\,dF(p)\\
 &\qquad\le H\!\left(\int_{s_1}^1p\,dF(p)+\int_{s_0}^1(1-p)\,dF(p)\right).
 \end{aligned}
 \label{eq:mixedfeasibility}
\end{equation}
\end{enumerate}
\end{theorem}

We call $Q$ \emph{jointly feasible} when it satisfies \eqref{eq:incentive-moment} and \eqref{eq:mixedfeasibility}.

The incentive characterization is an envelope argument with boundary restrictions imposed by voluntary actions. Truthful reporting together with obedience for each type at its truthful report also rules out a false report followed by reversed advice: these conditions imply $x(r)\ge y(r)$, making reversal no better than one of the two constant actions. Lemma~\ref{lem:a-ic} gives the argument. The joint implementation claim adds a different requirement: one recommendation kernel must realize all agents' interim probabilities simultaneously. It specializes generalized reduced-form feasibility \citep{Border1991,Border2007,CheKimMierendorff2013} to belief--state pairs. Monotonicity of $x_Q$ and $y_Q$ makes the two-cutoff groups sufficient. Checking only $s_1=s_0=0$, or $S+R\le k$, would miss the restrictions for smaller groups.

\subsection{When optimal advice displaces suitable agents}

The representation separates the exposure needed for incentives from its expected incidence. Applying Tonelli's theorem to \eqref{eq:participation-welfare} and \eqref{eq:Qrepresentation} gives
\begin{equation}
 S=n\int T_1\,dQ,\qquad
 R=\frac{nu_+}{u_-}\int vT_0\,dQ,\qquad
 W=nu_+\int(T_1-vT_0)\,dQ=n\int U\,dF
 \label{eq:welfareQ}
\end{equation}
for a symmetric rule represented by $Q$. We write these quantities as $S(Q)$, $R(Q)$, and $W(Q)$ when emphasizing their dependence on $Q$.

In the expression for $R$, $T_0(\tau)$ is the probability that an agent is unsuitable and has a belief above $\tau$. Expected harmful use therefore depends on the product $v(\tau)T_0(\tau)$, which tends to zero as $\tau\uparrow1$ because $v(\tau)T_0(\tau)\le1-F(\tau)$. This explains why concentrating exposure at high beliefs can reduce its incidence even though higher thresholds require more exposure conditional on the unsuitable state.

Symmetrization preserves $(S,R)$ and hence the advisor's value, even when $\Phi$ is nonlinear. The advisor therefore maximizes $\Phi(S(Q),R(Q))$ over jointly feasible threshold measures $Q$. The same feasible set applies to every objective.

The key to Theorem~\ref{thm:distortion}(ii) is a conflict between concentrating exposure and preserving beneficial use. Every maximizer has an outcome-equivalent canonical direct implementation satisfying
\begin{equation}
 y_i(1-)=1\qquad\text{for every agent }i.
 \label{eq:optimal-exposure}
\end{equation}
This is a conditional use probability, not the frequency of harmful use among confident agents. In contrast, preserving $S=S^{\mathrm{FB}}$ requires
\begin{equation}
 y_i(1-)\le b:=\Pr\{\Bin(n-1,\mu)\le k-1\}<1.
 \label{eq:bceiling}
\end{equation}
Indeed, an unsuitable agent can then use the resource only when fewer than $k$ other agents are suitable (Lemma~\ref{lem:b-displacement}). These incompatible endpoint requirements explain the strict sacrifice of beneficial use.

The reason is that an exposure probability below one leaves room to change which beliefs bear the screening burden. For reports below a high cutoff, move the offered option slightly toward a common option with no unsuitable exposure. Give that common option a suitable-state use probability just below the uniform direct-assignment probability $g$. This change reduces exposure over most of the belief distribution and leaves capacity slack. Above the cutoff, keep suitable-state use unchanged and increase unsuitable exposure by the amount needed to preserve reporting incentives. These agents are so rarely unsuitable that the additional harmful use can be made small relative to the reduction below the cutoff.

This argument must also respect capacity in each realized group, including groups concentrated near certainty. Lemma~\ref{lem:b-endpoint} verifies every subset capacity inequality for a sufficiently small change. The proof of Theorem~\ref{thm:distortion} then chooses the common option and cutoff to improve $\Phi$. Preservation prevents the adjustment from reaching an optimum: once an agent's unsuitable-state use probability exceeds $b$, its use must sometimes displace a suitable user. This is why minimizing harm can require giving up beneficial use.

\section{Optimal advice in a three-agent example}
\label{sec:example}

Consider three agents, a ceiling of two uses, uniform beliefs, and gains and losses normalized to one. Table~\ref{tab:uniform-welfare} compares welfare-maximizing advice with the least harmful way to preserve beneficial use and with observed-belief first best. Relative to the best preserving policy, the welfare optimum reduces beneficial use by about $0.06\%$ and harmful use by about $34\%$. These changes amount to about $0.00083$ fewer beneficial uses and $0.00742$ fewer harmful uses per group. Because $W=S-R$, the larger absolute reduction in harm raises welfare.

\begin{table}[ht]
\centering
\caption{Three agents and a two-use ceiling: use and welfare for $u_+=u_-=1$, rounded to five decimal places.}
\label{tab:uniform-welfare}
\begin{tabular}{lccc}
\toprule
Policy & $S$ & $R$ & $W=S-R$\\
\midrule
Minimum-harm preserving policy & $1.37500$ & $0.02192$ & $1.35308$\\
Welfare-maximizing policy & $1.37417$ & $0.01450$ & $1.35967$\\
Observed-belief first best & $1.37500$ & $0.00000$ & $1.37500$\\
\bottomrule
\end{tabular}
\end{table}

\Needspace{12\baselineskip}
The following menu achieves the welfare maximum.

\begin{example}[Welfare-maximizing advice with a two-use ceiling]
\label{ex:uniform}
Let beliefs be uniform, $n=3$, $k=2$, and $u_+=u_-=1$.
Write $x(p)$ and $y(p)$ for an agent's use probabilities conditional on belief $p$ and on the suitable and unsuitable states, respectively.
A welfare-maximizing policy offers the same suitable-state use probability and no unsuitable exposure to reports below $\alpha\simeq .889527$:
\[
 x(p)\simeq .89562,\qquad y(p)=0\qquad(p<\alpha).
\]
Both probabilities jump at $\alpha$, rise between $\alpha$ and $\beta=(1+\alpha)/2\simeq .944764$, jump again at $\beta$, and then rise toward one. Appendix~\ref{app:uniform-example} constructs the menu and proves optimality over all feasible mechanisms.
\end{example}

The menu is implemented through private binary recommendations that jointly respect capacity; agents do not receive separate, independently drawn advice.

The preserving policy in Table~\ref{tab:uniform-welfare} minimizes $R$ subject to $S=S^{\mathrm{FB}}$ and always gives suitable agents priority. Among them, reports above $\tau\simeq .83929$ have priority in decreasing order; lower reports are pooled with suitable-state use probability about $.88475$. Recommendations to use in the unsuitable state occupy only the remaining places and begin at $\tau$. Appendix~\ref{app:uniform-preserving} provides the exact menu, an implementing selection procedure, and the proof of minimum harm among all preserving mechanisms.

\begin{figure}[htbp]
\centering
\includegraphics[width=\textwidth]{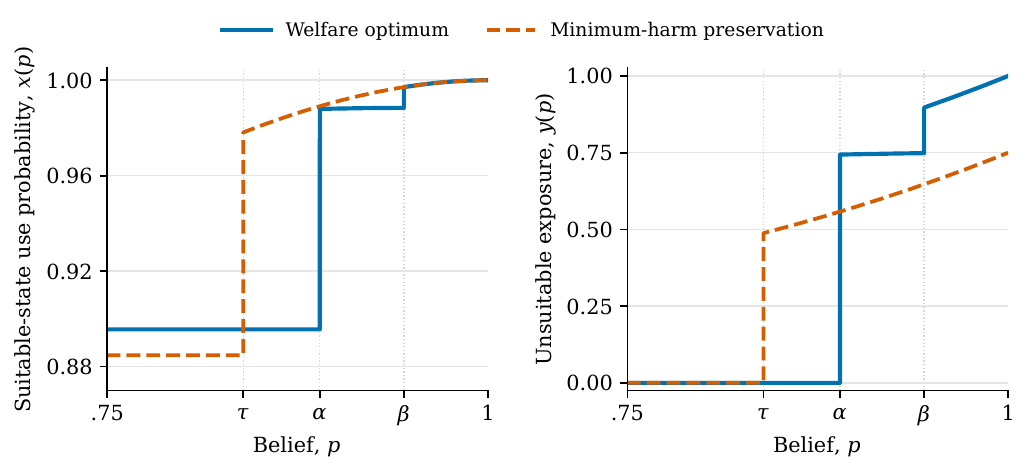}
\caption{Suitable-state use probability and unsuitable exposure in the three-agent example. Solid curves give welfare-maximizing advice; dashed curves give minimum-harm preservation. Both menus are constant for $p<.75$, which is omitted. The suitable-state panel also magnifies its vertical range near one. The jumps occur at $\alpha,\beta$ for the optimum and at $\tau$ for preservation.}
\label{fig:uniform-menus}
\end{figure}

Figure~\ref{fig:uniform-menus} shows that the welfare optimum delays the onset of unsuitable exposure to a higher belief. Its exposure approaches $y(1-)=1$, whereas the minimum-harm preserving policy reaches $y(1-)=3/4$, the preservation bound in \eqref{eq:bceiling}. Both menus respect the same realized ceiling; their difference illustrates how optimal screening trades a small amount of beneficial use for a larger reduction in harmful use.

\section{When coordination is feasible}
\label{sec:feasibility}

We now ask when coordination is possible as the incentive to use the resource changes. Fix $F,n,k$ and the loss $u_->0$, and increase the benefit $u_+$. We ask when private beliefs permit preservation of $S^{\mathrm{FB}}$, when they permit any coordination, and when observed beliefs permit coordination. In the last case, feasibility also permits $(S^{\mathrm{FB}},0)$. This comparison concerns the existence of a truthful and obedient policy, complementing the preceding comparison of attainable outcomes. Feasibility depends on payoffs only through $u_+/u_-$.

\begin{proposition}[Coordination with private and observed beliefs]
\label{prop:feasible-costs}
Fix $F,n,k$ and $u_->0$ as in the model, and vary $u_+>0$.
There are finite cutoffs $0\le\bar u_+^{\mathrm{FB}}\le\bar u_+\le\bar u_+^{\mathrm{obs}}$ such that:
\begin{enumerate}[label=(\roman*),leftmargin=*]
\item A feasible private-belief mechanism exists exactly for $0<u_+\le\bar u_+$.
\item A feasible private-belief mechanism with $S=S^{\mathrm{FB}}$ exists exactly for $0<u_+\le\bar u_+^{\mathrm{FB}}$.
\item If the advisor also privately observes $\mathbf p$, an obedient policy respecting capacity exists exactly for $0<u_+\le\bar u_+^{\mathrm{obs}}$, and it can attain $(S,R)=(S^{\mathrm{FB}},0)$.
\end{enumerate}
A zero cutoff means an empty range.
If $\bar u_+>0$, then $0<\bar u_+^{\mathrm{FB}}<\bar u_+<\bar u_+^{\mathrm{obs}}$.
\end{proposition}

The upper tail of beliefs determines whether these ranges are nonempty.
Using the suitable-state tail mass $T_1$ from \eqref{eq:tail-masses}, define
\begin{equation}
 I_k(F):=\int_0^1\frac{T_1(t)^k}{(1-t)^2}\,dt\in(0,\infty].
 \label{eq:integrability}
\end{equation}
Private-belief coordination is possible at some positive benefit, so that $\bar u_+>0$, if and only if $I_k(F)<\infty$.
Appendix~\ref{app:private-feasibility} proves this criterion, the two private-belief intervals, and their strict separation.
The condition $I_k(F)<\infty$ has a direct interpretation:
\[
 I_k(F)=\mu^k\E\!\left[
       \min_{1\le j\le k}\frac{p_j}{1-p_j}
       \,\middle|\,\theta_1=\cdots=\theta_k=1\right].
\]
Conditional on $k$ independent agents being suitable, their lowest belief
odds must have a finite mean. This limits how often all $k$ beliefs are
close to certainty without bounding them away from one.

When $I_k(F)<\infty$, Table~\ref{tab:feasibility-ranges} displays the three nested ranges, including their endpoints. Throughout the private-belief feasible range, harmful use has a strictly positive minimum. Even in the lowest range, where preservation is feasible, it is suboptimal for every objective in Theorem~\ref{thm:distortion}(ii).

\begin{table}[htbp]
\centering
\caption{Feasible benefits when $I_k(F)<\infty$, holding $u_-$ fixed}
\label{tab:feasibility-ranges}
\small
\renewcommand{\arraystretch}{1.2}
\begin{tabularx}{\textwidth}{@{}Xccc@{}}
\toprule
& \multicolumn{2}{c}{Private beliefs} & Observed beliefs\\
\cmidrule(lr){2-3}\cmidrule(l){4-4}
Benefit $u_+$ & $S=S^{\mathrm{FB}}$ & Any coordination & $(S^{\mathrm{FB}},0)$\\
\midrule
$0<u_+\le\bar u_+^{\mathrm{FB}}$ & Yes & Yes & Yes\\
$\bar u_+^{\mathrm{FB}}<u_+\le\bar u_+$ & No & Yes & Yes\\
$\bar u_+<u_+\le\bar u_+^{\mathrm{obs}}$ & No & No & Yes\\
$u_+>\bar u_+^{\mathrm{obs}}$ & No & No & No\\
\bottomrule
\end{tabularx}
\end{table}

The observed-belief cutoff has the explicit expression
\begin{equation}
 \bar u_+^{\mathrm{obs}}
 =u_-\inf_{0\le t<1}
       \frac{nT_0(t)}{nT_1(t)-H(T_1(t))}\in[0,\infty).
 \label{eq:observed-cutoff}
\end{equation}
Appendix~\ref{app:observed-beliefs} proves this formula and the strict comparison with private beliefs.
If $I_k(F)=\infty$, both private-belief cutoffs are zero, while the observed-belief range is nonempty precisely when
\begin{equation}
 \sup_{0\le t<1}\frac{T_1(t)^k}{1-t}<\infty.
 \label{eq:observed-tail}
\end{equation}
This is equivalent to $1-F(t)=O((1-t)^{1/k})$ as $t\uparrow1$. Thus obedience alone imposes a bounded-tail condition; elicitation imposes the stronger integrability condition \eqref{eq:integrability}. The uniform example below shows that the distinction can determine whether any coordination is possible.

\paragraph{Why the ranges differ.}
The criterion $I_k(F)<\infty$ separates whether coordination is possible at some positive benefit from how large the benefit-to-loss ratio can be. Existence depends on $k$ and the upper tail of $F$, but not on $n>k$; the critical ratio can depend on $n$. For necessity, when $k+1$ agents are suitable and all have beliefs above $t$, at least one must forgo use. For symmetric canonical interims, monotonicity therefore implies
\[
 (k+1)T_1(t)[1-x(t)]\ge T_1(t)^{k+1}.
\]
Incentives limit these nonuse probabilities through the identity
\[
 \int v\,dQ=\int_0^1\frac{1-x(t)}{(1-t)^2}\,dt\le\frac{u_-}{u_+},
\]
proved in Lemma~\ref{lem:a-ic}. Combining the two displays gives $I_k(F)<\infty$: excessive crowding among nearly certain agents cannot be reconciled with their willingness to forgo use.

The construction for part (ii) selects suitable agents in decreasing report order and uses spare places to supply the exposure needed for truthful reports. It guarantees preservation whenever
\begin{equation}
 \frac{u_+}{u_-}\binom{n-1}{k}I_k(F)
 \le\frac{k-S^{\mathrm{FB}}}{n(1-\mu)}.
 \label{eq:preservation-sufficient}
\end{equation}
The left side bounds the highest unsuitable-state use probability required by the construction. The right side is an unsuitable agent's selection probability when spare places are filled uniformly among unsuitable agents. When $I_k(F)<\infty$, this sufficient bound makes the preserving range nonempty. In that case, its upper endpoint is nevertheless strictly below the limit of private-belief feasibility: at the latter limit, every feasible mechanism must exhaust unsuitable exposure at the highest beliefs, whereas preservation requires that exposure to stay below one.

Formula~\eqref{eq:observed-cutoff} gives an explicit necessary and sufficient condition for observed-belief coordination. Equivalently, a feasible observed-belief policy exists exactly when
\begin{equation}
 u_+\big[nT_1(t)-H(T_1(t))\big]\le nu_-T_0(t)
 \qquad(0\le t<1).
 \label{eq:observed-feasibility}
\end{equation}
Write $(z)_+=\max\{z,0\}$ for the positive part of $z$. Among suitable agents with beliefs above $t$, capacity leaves at least
\[
 nT_1(t)-H(T_1(t))=\E\big[(\Bin(n,T_1(t))-k)_+\big]
\]
unserved on average. The condition says that their forgone benefit must be covered by the unsuitable-state loss that nonuse can avoid. Every upper tail must satisfy this comparison; the aggregate inequality at $t=0$ need not suffice. Lemma~\ref{lem:observed-beliefs} proves sufficiency and attainment of first best. This condition and improvement do not require full support.

Finally, for any privately feasible policy, removing unsuitable exposure and filling suitable places creates enough obedience slack to accommodate a strictly larger benefit. Appendix~\ref{app:observed-beliefs} verifies that the slack is uniform across beliefs. When $\bar u_+>0$, applying the improvement at the attained private-belief boundary gives $\bar u_+<\bar u_+^{\mathrm{obs}}$. Observation thus enlarges both the attainable outcome set and the range of incentives permitting coordination.

\subsection{Uniform beliefs: the role of a higher usage ceiling}

For uniform beliefs, $T_1(t)=(1-t^2)/2$, so
\[
 I_k(F)=2^{-k}\int_0^1(1-t)^{k-2}(1+t)^k\,dt<\infty
 \quad\Longleftrightarrow\quad k\ge2.
\]
It follows that, for any $n>k$, a ceiling of one use makes private-belief coordination impossible at every positive benefit-to-loss ratio, whereas a ceiling of two or more permits preservation at all sufficiently small positive ratios.
Allowing an additional use changes whether truthful and obedient advice is possible at all; it does not eliminate its screening cost.

The difference is especially transparent with two agents and a ceiling of one use. Here $T_0(t)=(1-t)^2/2$ and $I_1(F)=\infty$, while \eqref{eq:observed-cutoff} gives
\[
 \frac{\bar u_+^{\mathrm{obs}}}{u_-}
 =\inf_{0\le t<1}\frac{4}{(1+t)^2}=1.
\]
Thus observing beliefs permits first best exactly for $0<u_+\le u_-$, whereas private-belief coordination is impossible at every positive ratio.

An observed-belief policy implementing this entire range selects the sole suitable agent if there is one, the one with the higher belief if both are suitable, and nobody otherwise. A belief-$p$ agent then has
\[
 x(p)=1-\int_p^1r\,dr=\frac{1+p^2}{2},\qquad y(p)=0.
\]
A recommendation to use the resource reveals suitability. Nonuse pools being unsuitable with being suitable but less confident than the other suitable agent. Since $u_+\le u_-$, it is obedient:
\[
 u_+p\frac{1-p^2}{2}\le u_-(1-p).
\]
Ranking by reported beliefs would fail: with no unsuitable exposure, every type would prefer the higher use probability offered to a higher report. With three agents and a ceiling of two uses, Example~\ref{ex:uniform} establishes private-belief feasibility even at equal gains and losses, and shows the resulting welfare cost.

\subsection{Why finite reports cannot coordinate use}

The preceding feasibility results allow unrestricted reports. Could the advisor instead coordinate use using finitely many reported categories, such as low and high confidence?

\begin{proposition}[Finite reports]
\label{prop:finite-reports}
If at least $k$ agents have finite report spaces, no Bayesian Nash equilibrium can satisfy $\sum_i a_i\le k$ almost surely, even with unrestricted standard Borel message spaces.
\end{proposition}

Equivalently, every feasible mechanism requires infinite report spaces for at least $n-k+1$ agents; when $k=1$, this applies to every agent. The conclusion allows unrestricted messages from the advisor. It concerns reports sent to the advisor. Whenever coordination is feasible, the advisor can reply with binary recommendations.

Finite reporting must pool some beliefs arbitrarily close to certainty. For each of $k$ agents with finite report spaces, choose a report compatible with arbitrarily high beliefs. On the positive-probability event that these agents are suitable and send those reports, there remains positive probability that all $k$ use the resource, even conditional on their messages and the other agents' actions. Hard capacity therefore forces every outside agent to forgo use almost surely on that event. The event is independent of an outside agent's own belief and state, so its suitable-state nonuse probability would be bounded away from zero even as its belief approaches one. This contradicts the obedience inequality \eqref{eq:observed-obedience}, which is necessary under private beliefs as well. Appendix~\ref{app:reporting} makes the conditional argument precise.

In their finite-type framework, \citet[Proposition~6]{BergemannMorris2019} show that elicitation does not enlarge the set of attainable decision rules for one receiver with two states and two actions. Their two-receiver example shows that elicitation can matter even with binary actions and payoffs that do not depend on others' actions \citep[Appendix, Example~1]{BergemannMorris2019}. Proposition~\ref{prop:finite-reports} identifies a requirement for feasibility under a realized usage ceiling and full-support beliefs: at least $n-k+1$ agents must have infinite report spaces. This restriction applies to every feasible outcome, independently of the advisor's objective. \citet{BergemannHeumannMorris2026} jointly choose buyer information and a quality--price menu, obtaining finite optimal menus; our requirement of infinitely many possible reports concerns pre-existing private beliefs and voluntary use under a realized ceiling.

\section{Conclusion}

Under full-support beliefs, harmful use has a strictly positive minimum whenever private-belief coordination is feasible. Optimal advice sacrifices beneficial use for every objective considered here, including welfare maximization and harm minimization. Whenever preservation is feasible, even an arbitrarily small departure can improve both objectives. In these feasible environments, observing beliefs allows voluntary advice to attain $(S^{\mathrm{FB}},0)$ and strictly improves the advisor's value. The efficiency loss therefore comes from elicitation.

The assumptions locate the friction. A realized usage ceiling links otherwise independent decisions, and beliefs approaching certainty constrain which nonuse recommendations can be obeyed. Whenever private-belief coordination is possible, its range of benefits from use strictly contains the preserving range and is strictly contained in the observed-belief range. The cost need not survive when beliefs are bounded away from certainty: under \eqref{eq:zero-unsuitable}, report-free advice attains first best. Under the maintained full-support assumption, however, exact coordination requires infinite report spaces for at least $n-k+1$ agents, even though binary recommendations suffice. These are limits on implementing a given usage ceiling through voluntary advice.

\newpage

\appendix
\section{Proofs of implementation}

\subsection{Proof of Theorem~\ref{thm:implementation}}

We first reduce arbitrary equilibria to direct recommendations with reports in $(0,1)$, then complete the interims to canonical rules and symmetrize while preserving $(S,R)$. The final two lemmas characterize individual incentives and joint capacity, respectively.

\subsubsection*{General communication and the revelation reduction}
\label{app:general-communication}

A \emph{report-and-message mechanism} is $\Pi=((\mathcal R_i,\mathcal M_i)_{i=1}^n,\pi)$, where $\mathcal R_i$ and $\mathcal M_i$ are nonempty standard Borel report and message spaces. For report and message random vectors $\mathbf r=(r_i)_{i=1}^n$ and $\mathbf m=(m_i)_{i=1}^n$, the probability kernel $\pi(d\mathbf m\mid\mathbf r,\boldsymbol\theta)$ specifies the message distribution at every report--state profile.

The advisor publicly commits to $\Pi$ before beliefs and states are realized, as in the direct formulation. After observing their respective information, agents simultaneously submit private reports. The advisor draws $\mathbf m\sim\pi(\cdot\mid\mathbf r,\boldsymbol\theta)$ and privately sends $m_i$ to agent $i$. Agents then act simultaneously. Each observes only its own type, report, message, and private randomization; no further diagnostic information is available before acting. Both actions remain available after every report and message. The advisor cannot use transfers, compel actions, deny access, or physically ration use.

A strategy $\sigma_i=(\sigma_i^{\mathrm r},\sigma_i^{\mathrm a})$ consists of measurable probability kernels $\sigma_i^{\mathrm r}(dr_i\mid p_i)$ and $\sigma_i^{\mathrm a}(da_i\mid p_i,r_i,m_i)$. Private randomization is independent across agents and of the primitives; the advisor randomizes separately. Strategies are behavioral with perfect recall. Write $\Pr_{\Pi,\sigma}$ and $\E_{\Pi,\sigma}$ for the joint law induced by mechanism $\Pi$ and strategy profile $\sigma$. A Bayesian Nash equilibrium satisfies, for every agent $i$ and alternative strategy $\widehat\sigma_i$,
\begin{equation}
 \E_{\Pi,\sigma}[u_i(a_i,\theta_i)]
 \ge
 \E_{\Pi,(\widehat\sigma_i,\sigma_{-i})}[u_i(a_i,\theta_i)],
 \label{eq:bayesiannash}
\end{equation}
where $\sigma_{-i}$ denotes the other agents' strategies. Deviations may depend on type and change both the report and the subsequent action rule, including following a false report. Equilibrium optimality is evaluated under $F$.

A mechanism--equilibrium pair $(\Pi,\sigma)$ is feasible when $\sum_i a_i\le k$ almost surely under $\Pr_{\Pi,\sigma}$. It implements the corresponding joint law of $(\mathbf p,\boldsymbol\theta,\mathbf a)$. The capacity requirement concerns this equilibrium alone and does not restrict deviations. A direct recommendation mechanism has $\mathcal R_i=(0,1)$ and $\mathcal M_i=\{0,1\}$. The notation $\Pr_\pi,\E_\pi$ in the main text abbreviates the law and expectation under $(\Pi,\sigma^{\mathrm{TO}})$.

\begin{lemma}
\label{lem:a-revelation}
For every feasible mechanism--equilibrium pair $(\Pi,\sigma)$, there is a direct recommendation kernel $\pi^{\mathrm{raw}}$ with report space $(0,1)$ for which $\sigma^{\mathrm{TO}}$ is a Bayesian Nash equilibrium and the joint law of $(\mathbf p,\boldsymbol\theta,\mathbf a)$ is unchanged. This conclusion does not require Assumption~\ref{ass:belief-support}.
\end{lemma}

\begin{proof}
Fix a feasible pair $(\Pi,\sigma)$, with kernel $\pi$. At a
report--state profile $(\mathbf r,\boldsymbol\theta)\in(0,1)^n\times\{0,1\}^n$, simulate
\[
 \widetilde r_i\sim\sigma_i^{\mathrm r}(\cdot\mid r_i),\qquad
 \widetilde{\mathbf m}\sim\pi(\cdot\mid\widetilde{\mathbf r},\boldsymbol\theta),
 \qquad
 m_i\sim\sigma_i^{\mathrm a}(\cdot\mid r_i,\widetilde r_i,\widetilde m_i),
\]
using independent reporting and action draws across agents. Send only $m_i$.
Composition of probability kernels on standard Borel spaces defines a
kernel $\pi^{\mathrm{raw}}$. Truthful obedience reproduces the
original joint law of $(\mathbf p,\boldsymbol\theta,\mathbf a)$.

Consider any deviation $\widehat\sigma_i$ in this direct mechanism. An original-game
deviator first draws an auxiliary direct report $d_i\in(0,1)$, then submits
$\widetilde r_i$, according to
\[
 d_i\sim\widehat\sigma_i^{\mathrm r}(\cdot\mid p_i),\qquad
 \widetilde r_i\sim\sigma_i^{\mathrm r}(\cdot\mid d_i).
\]
After receiving $\widetilde m_i$, it draws
\[
 \widetilde a_i\sim\sigma_i^{\mathrm a}(\cdot\mid d_i,\widetilde r_i,\widetilde m_i),
 \qquad
 a_i\sim\widehat\sigma_i^{\mathrm a}(\cdot\mid p_i,d_i,\widetilde a_i).
\]
This is an admissible behavioral deviation: conditional on
$(p_i,\widetilde r_i)$, the retained auxiliary report is independent of
the state and original message, so its conditional law can be integrated
out of the action rule. The resulting payoff equals that from the direct
deviation. Original equilibrium optimality therefore rules out every joint
reporting--action deviation in the constructed mechanism. Truthful obedience
implements the original outcome.

\end{proof}

For a direct kernel, the pointwise recommendation-capacity condition is
\begin{equation}
 \pi\!\left(\left\{\mathbf m:\sum_i m_i\le k\right\}\mid\mathbf r,\boldsymbol\theta\right)=1
 \quad\text{for every }(\mathbf r,\boldsymbol\theta)\in(0,1)^n\times\{0,1\}^n.
 \label{eq:pointwiserecommendationcapacity}
\end{equation}

\begin{lemma}
\label{lem:a-completion}
Under Assumption~\ref{ass:belief-support}, suppose a direct recommendation kernel $\pi^{\mathrm{raw}}$ makes $\sigma^{\mathrm{TO}}$ a Bayesian Nash equilibrium and satisfies $\Pr_{\pi^{\mathrm{raw}}}\{\sum_i m_i\le k\}=1$. There is a direct kernel $\pi^{\mathrm{canon}}$ inducing the same joint law of $(\mathbf p,\boldsymbol\theta,\mathbf a)$ under $\sigma^{\mathrm{TO}}$ such that:
\begin{enumerate}[label=(\roman*),leftmargin=*]
\item its interims $(x_i,y_i)$ are nondecreasing and right-continuous and satisfy \eqref{eq:jointIC} for every $i$ and every $p,r\in(0,1)$;
\item it satisfies \eqref{eq:pointwiserecommendationcapacity} at every report--state profile.
\end{enumerate}
In particular, $\sigma^{\mathrm{TO}}$ is a Bayesian Nash equilibrium of the completed mechanism.
\end{lemma}

\begin{proof}
Write $(x_i^{\mathrm{raw}},y_i^{\mathrm{raw}})$ for the raw interims.
For each agent, choose reports whose option pairs are countable and dense
in their range. For each chosen report and each of the four continuation
rules, deviating where its utility gain $\Delta U_i(p)$ is positive gives
$\int[\Delta U_i(p)]_+dF(p)=0$ by equilibrium. Intersecting these
full-measure sets and using continuity in the option pair shows that the
raw interims satisfy \eqref{eq:jointIC} for every report and every true
type in a common Borel set $\mathcal T\subset(0,1)$ of full $F$-measure.
This common-full-measure argument does not use Assumption~\ref{ass:belief-support}.

Full support makes $\mathcal T$ dense in $(0,1)$. For $p>r$ in $\mathcal T$, the
two report-and-obey comparisons give
\[
 u_+v(r)\Delta x\le u_-\Delta y\le u_+v(p)\Delta x,
 \qquad
 (\Delta x,\Delta y)
 =(x_i^{\mathrm{raw}}(p)-x_i^{\mathrm{raw}}(r),
   y_i^{\mathrm{raw}}(p)-y_i^{\mathrm{raw}}(r)).
\]
Both differences are nonnegative. Define the right-continuous extensions
\[
 (x_i(r),y_i(r))=
 \lim_{\substack{t\downarrow r\\t\in\mathcal T}}
 (x_i^{\mathrm{raw}}(t),y_i^{\mathrm{raw}}(t)).
\]
Limits from $\mathcal T$ in both arguments preserve \eqref{eq:jointIC}
at every $p,r$. On $\mathcal T$, raw values differ from their extensions
only at jumps, a countable set; atomlessness makes these changes null.

Replace recommendations by $\mathbf0$ on the Borel, $\rho^n$-null set
where the raw kernel violates capacity, obtaining $\pi^{\mathrm{cap}}$.
Because interior beliefs give both states positive probability,
$\rho^n$ is equivalent to $F^n$ times counting measure. Fubini implies
that this replacement changes each agent's state-contingent interims
only on a Borel $F$-null set of own reports. Let $\mathcal N$ be the union of
all those exceptional reports, the complement of $\mathcal T$, and all
canonical jumps; it is a Borel $F$-null set. Outside $\mathcal N$, $\pi^{\mathrm{cap}}$ has the
canonical interims.

If every report lies outside $\mathcal N$, retain $\pi^{\mathrm{cap}}$.
Otherwise let $i_*=\min\{i:r_i\in\mathcal N\}$ and use the kernel
\[
 z_*\delta_{\mathbf e_{i_*}}+(1-z_*)\delta_{\mathbf0},\qquad
 z_*=\theta_{i_*}x_{i_*}(r_{i_*})+
       (1-\theta_{i_*})y_{i_*}(r_{i_*}),
\]
where $\mathbf e_i$ has only coordinate $i$ equal to one. This measurable
completion satisfies pointwise capacity because $k\ge1$. A unilateral
exceptional report faces truthful opponents outside $\mathcal N$ almost
surely and receives its canonical option; other reports retain their
interims. All changes are null under truthful inputs, so the completed
kernel preserves the outcome and has the prescribed interims at every
report.
\end{proof}

\begin{lemma}
\label{lem:a-symmetry}
Let $\pi$ be a direct kernel with canonical interims $(x_i,y_i)_{i=1}^n$ satisfying \eqref{eq:jointIC} and with pointwise capacity \eqref{eq:pointwiserecommendationcapacity}. There is a symmetric direct kernel $\bar\pi$ with the same $(S,R)$, the same capacity property, and common interims
\[
 \bar x(p)=\frac1n\sum_{i=1}^n x_i(p),\qquad
 \bar y(p)=\frac1n\sum_{i=1}^n y_i(p).
\]
These interims are canonical and satisfy \eqref{eq:jointIC}.
\end{lemma}

\begin{proof}
For a permutation $\chi$ of the agent indices and an event $\mathcal C\subseteq\{0,1\}^n$, write $(\chi\mathbf r)_i=r_{\chi^{-1}(i)}$, similarly for states and recommendations, and $\chi\mathcal C=\{\chi\mathbf m:\mathbf m\in\mathcal C\}$. Apply a uniformly random permutation, use $\pi$, and undo the permutation without revealing it:
\[
 \bar\pi(\mathcal C\mid\mathbf r,\boldsymbol\theta)
 =\frac1{n!}\sum_\chi\pi(\chi\mathcal C\mid\chi\mathbf r,\chi\boldsymbol\theta).
\]
Exchangeability gives the stated common interims and preserves $S,R$. Averaging preserves right-continuity, monotonicity, and the four affine incentive inequalities in \eqref{eq:jointIC}. Relabeling preserves the number of recommendations to use, so pointwise capacity is unchanged.
\end{proof}

\begin{lemma}
\label{lem:a-ic}
A right-continuous pair $(x,y):(0,1)\to[0,1]^2$ is incentive compatible if and only if it has representation \eqref{eq:Qrepresentation} for a probability measure $Q$ on $[0,1)$ satisfying \eqref{eq:incentive-moment}. In that case, $Q=x(0+)\delta_0+dx$, $x$ and $y$ are nondecreasing, $x(1-)=1$, $y(0+)=0$, $y\le x$, and
\begin{equation}
 \int_{[0,1)}v(\tau)\,Q(d\tau)
 =\int_0^1\frac{1-x(t)}{(1-t)^2}\,dt.
 \label{appimpl:tailcharge}
\end{equation}
For a symmetric direct mechanism with these interims, $(S,R,W)$ satisfy \eqref{eq:welfareQ}.
\end{lemma}

\begin{proof}
For $0<r<p<1$, the two report-and-obey comparisons give
\[
 u_+v(r)[x(p)-x(r)]\le u_-[y(p)-y(r)]
 \le u_+v(p)[x(p)-x(r)].
\]
Thus $x$ and $y$ are nondecreasing. Summing these bounds along a
partition of any compact subinterval and letting its mesh tend to zero
gives the Stieltjes identity $u_-\,dy=u_+v\,dx$, since $v$ is continuous
there. The argument includes atoms and singular continuous parts.
Obedience gives
\[
 u_-(1-p)y(p)\le u_+p,\qquad
 u_+p[1-x(p)]\le u_-(1-p)[1-y(p)]\le u_-(1-p).
\]
Hence $y(0+)=0$ and $x(1-)=1$. The probability measure
$Q=x(0+)\delta_0+dx$ therefore has representation
\eqref{eq:Qrepresentation}, and $y(1-)\le1$ gives the odds bound.

Conversely, take such a $Q$. Its represented pair is right-continuous,
lies in $[0,1]^2$, and satisfies
\[
 U(p;r)=u_+\int_{[0,r]}\frac{p-\tau}{1-\tau}\,Q(d\tau).
\]
The integrand changes sign at $p$, so truth-telling maximizes this
expression and $U(p)\ge0$. The odds bound gives
\[
 u_-[1-y(p)]\ge u_+\int_{(p,1)}v\,dQ
 \ge u_+v(p)[1-x(p)],
\]
or $U(p)\ge U^1(p)$. Multiplying these two obedience inequalities
at $r$ by $1-x(r)$ and $x(r)$, respectively, gives
$y(r)[1-x(r)]\le x(r)[1-y(r)]$, hence $y(r)\le x(r)$.
Reversal after any report then satisfies
\[
 U^1(p)-U(p;r)
 =(1-y(r))U^1(p)-u_+p[x(r)-y(r)]
 \le\max\{0,U^1(p)\}.
\]
This proves all four comparisons in \eqref{eq:jointIC}.

Finally, Tonelli's theorem gives \eqref{appimpl:tailcharge} from
$v(\tau)=\int_0^\tau(1-t)^{-2}dt$ and $Q((t,1))=1-x(t)$.
Applying the same theorem to \eqref{eq:Qrepresentation} yields
$S/n=\int T_1\,dQ$ and $R/n=(u_+/u_-)\int vT_0\,dQ$,
and hence \eqref{eq:welfareQ}.
\end{proof}

For the next lemma and subsequent proofs, define the \emph{stop-loss transforms}
\begin{equation}
 \begin{aligned}
 L_Q(\ell)&=\int_0^1\!\left[p(x_Q(p)-\ell)_+
              +(1-p)(y_Q(p)-\ell)_+\right]dF(p),\\
 L(\ell)&=\max_{0\le\eta\le1}\left\{\frac{H(\eta)}n-\ell\eta\right\},
 \qquad 0\le\ell\le1.
 \end{aligned}
 \label{eq:stoploss-transforms}
\end{equation}

\begin{lemma}
\label{lem:a-capacity}
Fix a probability measure $Q$ on $[0,1)$ satisfying \eqref{eq:incentive-moment}. The following are equivalent:
\begin{enumerate}[label=(\roman*),leftmargin=*]
\item $(x_Q,y_Q)$ is jointly implementable under hard capacity;
\item the subset capacity inequalities \eqref{eq:mixedfeasibility} hold for every $s_1,s_0\in[0,1]$;
\item the inequalities
\begin{equation}
 L_Q(\ell)\le L(\ell)\qquad\text{for every }\ell\in[0,1)
 \label{eq:stoploss}
\end{equation}
hold.
\end{enumerate}
When these conditions hold, an implementing direct kernel has the interims $(x_Q,y_Q)$ at every report and satisfies \eqref{eq:pointwiserecommendationcapacity}.
\end{lemma}

\begin{proof}
Write $x=x_Q$, $y=y_Q$. For $(p,\theta)\sim\rho$ and $\omega\sim\operatorname{Unif}[0,1]$, define
\[
 Z_Q(p,\theta)=\theta x(p)+(1-\theta)y(p),\qquad
 Z=\kappa(\omega),\qquad
 \kappa(\omega)=\Pr\{\Bin(n-1,1-\omega)\le k-1\}.
\]
Thus $L_Q(\ell)=\E_\rho(Z_Q-\ell)_+$. The variable $Z$ is the interim selection probability under the rule selecting the $k$ highest independent uniform scores. For $1\le k<n$, $\kappa$ is a continuous strictly increasing bijection of $[0,1]$, since
\[
 \kappa'(\omega)=(n-1)\binom{n-2}{k-1}(1-\omega)^{k-1}\omega^{n-k-1}>0
 \quad(0<\omega<1).
\]
Using $H'(\eta)=n\Pr\{\Bin(n-1,\eta)\le k-1\}$ and $H(0)=0$ gives
\begin{equation}
 \int_{1-\eta}^1\kappa(\omega)\,d\omega=\frac{H(\eta)}n
 \quad(0\le\eta\le1).
 \label{appimpl:rankidentity}
\end{equation}
Write $\mathbf1_A$ for the indicator of an event $A$. For any implementing kernel and Borel set $\mathcal B\subset(0,1)\times\{0,1\}$, capacity implies
\[
 \sum_i m_i\mathbf1_{\{(p_i,\theta_i)\in\mathcal B\}}
 \le\min\left\{k,\sum_i\mathbf1_{\{(p_i,\theta_i)\in\mathcal B\}}\right\}.
\]
Expectations give $n\int_{\mathcal B}Z_Q\,d\rho\le H(\rho(\mathcal B))$; the sets \eqref{eq:capacity-upper-sets} yield (ii).

Conversely, suppose (ii). By monotonicity, $\mathcal B_\ell=\{Z_Q>\ell\}$ is a union of two upper intervals, up to null endpoints. Applying (ii) gives
\[
 L_Q(\ell)=\int_{\mathcal B_\ell}Z_Q\,d\rho-\ell\rho(\mathcal B_\ell)
 \le\frac{H(\rho(\mathcal B_\ell))}n-\ell\rho(\mathcal B_\ell)\le L(\ell).
\]
By \eqref{appimpl:rankidentity} and monotonicity of $\kappa$,
\[
 L(\ell)=\sup_{0\le\eta\le1}\int_{1-\eta}^1[\kappa(\omega)-\ell]d\omega
 =\E(Z-\ell)_+.
\]
The inequalities comparing these expected positive parts are called \emph{stop-loss inequalities}. Including $\ell=0$, they characterize \emph{increasing convex order}: $\E\varphi(Z_Q)\le\E\varphi(Z)$ for every increasing convex function $\varphi$. For bounded nonnegative variables, the equivalence follows by approximating increasing convex functions with constants plus nonnegative linear terms and nonnegative combinations of $(z-\ell)_+$.

Under (iii), the increasing-convex-order form of Strassen's theorem \citep{Strassen1965,LeskelaVihola2017} therefore couples $Z_Q,Z$ with $\E[Z\mid Z_Q]\ge Z_Q$. Let $\Gamma(dz'\mid z)$ be a Borel conditional law of $Z$ given $Z_Q=z$. On the Borel set, null under the law of $Z_Q$, where $\int z'\Gamma(dz'\mid z)<z$, replace it by $\delta_z$. This preserves the second marginal and ensures
\begin{equation}
 \bar z(z):=\int z'\Gamma(dz'\mid z)\ge z\quad(z\in[0,1]).
 \label{appimpl:submartingale}
\end{equation}
At report--state profile $(\mathbf r,\boldsymbol\theta)$, put $z_i=Z_Q(r_i,\theta_i)$ and draw independent scores $\widetilde z_i\sim\Gamma(\cdot\mid z_i)$. Select the $k$ largest, breaking ties with independent uniform priorities and then indices. Retain each selected agent independently with probability $z_i/\bar z(z_i)$, or zero if $\bar z(z_i)=0$. This defines a Borel kernel with pointwise capacity.

For any fixed own report and state, truthful opponents' scores have the law of $Z$, so their transformed scores $\kappa^{-1}(\widetilde z_j)$ are independent uniforms. An own score $z'$ is consequently selected with probability
\[
 \Pr\{\Bin(n-1,1-\kappa^{-1}(z'))\le k-1\}=z',
\]
including at zero and one; ties with opponents have probability zero. Averaging gives pre-retention probability $\bar z(z)$ and final probability $z$. If $\bar z(z)=0$, \eqref{appimpl:submartingale} gives $z=0$ as required. The kernel thus implements the interims at every report; Lemma~\ref{lem:a-ic} supplies incentives, proving (i).
\end{proof}

\begin{proof}[Proof of Theorem~\ref{thm:implementation}]
Lemma~\ref{lem:a-ic} proves part (i), and Lemma~\ref{lem:a-capacity} gives the capacity characterization in part (ii). Averaging an implementing kernel over permutations as in Lemma~\ref{lem:a-symmetry} preserves its common interims and capacity while making the kernel symmetric. Conversely, any symmetric feasible implementation satisfies the same subset capacity inequalities.
\end{proof}

\Needspace{12\baselineskip}
\subsection{The role of extreme beliefs}

\begin{lemma}[The zero-unsuitable-use boundary]
\label{lem:a-support}
Without Assumption~\ref{ass:belief-support}, and with all other assumptions maintained, a feasible mechanism attains $R=0$ if and only if \eqref{eq:zero-unsuitable} holds. When it holds, the direct-assignment benchmark is implementable without reports.
\end{lemma}

\begin{proof}
Suppose a feasible mechanism has $R=0$ without Assumption~\ref{ass:belief-support}. Lemma~\ref{lem:a-revelation}, which does not use that assumption, gives a direct implementation. Its interims satisfy
$\sum_i\int(1-p)y_i(p)dF(p)=0$, so $y_i=0$ $F$-almost everywhere. Put $q_i=\int_0^1x_i\,dF$. Drawing a report independently from $F$ and obeying then gives type $p$ utility $u_+pq_i$. Making this deviation only if $x_i(p)<q_i$ would gain
$u_+\int p[q_i-x_i(p)]_+dF(p)$. Thus $x_i\ge q_i$ almost everywhere; integration gives equality. Always using the resource then implies
\[
 u_+p(1-q_i)\le u_-(1-p)\quad F\text{-almost everywhere}.
\]
Capacity gives $\mu\sum_iq_i=S\le n\mu g$. Averaging the preceding inequality therefore yields
$u_+p(1-g)\le u_-(1-p)$ almost everywhere, equivalently \eqref{eq:zero-unsuitable}.

Conversely, under \eqref{eq:zero-unsuitable}, recommend use to a uniformly chosen set of $\min\{k,\sum_i\theta_i\}$ suitable agents without reports. The interims are $(g,0)$. A recommendation to use reveals suitability; after a recommendation to forgo use the unnormalized gain from use is $u_+p(1-g)-u_-(1-p)\le0$ almost everywhere. Obedience is therefore an equilibrium attaining $(S,R)=(S^{\mathrm{FB}},0)$ with pointwise capacity.
\end{proof}

\subsection{Compactness}

\begin{lemma}[Compactness of the threshold feasible set]
\label{lem:a-compactness}
For every $u_+,u_->0$, the set of jointly feasible threshold measures,
viewed as probability measures on $[0,1]$, is weakly compact.
The maps $Q\mapsto S(Q)$ and $Q\mapsto R(Q)$ are continuous on this set.
\end{lemma}

\begin{proof}
Extend $v$ lower-semicontinuously by $v(1)=+\infty$. Probability measures
on $[0,1]$ with $u_+\int v\,dQ\le u_-$ form a weakly compact set by
Portmanteau; every such measure has zero mass at one. Substituting the
threshold representation into the capacity inequalities gives
\begin{equation}
 n\int\!\left[T_1(\max\{s_1,\tau\})
 +\frac{u_+}{u_-}v(\tau)T_0(\max\{s_0,\tau\})\right]Q(d\tau)
 \le H(T_1(s_1)+T_0(s_0)).
 \label{appimpl:compactconstraints}
\end{equation}
For each fixed cutoff pair, the integrand extends continuously by zero
at one: $F$ is atomless and
$v(\tau)T_0(\max\{s_0,\tau\})\le1-F(\tau)\to0$.
Each capacity inequality therefore defines a weakly closed set, so their
intersection with the set satisfying the odds bound is compact. The same argument
shows that $T_1$ and $vT_0$, extended by zero at one, are bounded and
continuous. Formula~\eqref{eq:welfareQ} proves continuity of $S$ and $R$.
\end{proof}

\section{Proof of the capacity cost of screening}
\label{app:capacity-cost}

\subsection{Proof of Theorem~\ref{thm:distortion}}

Throughout this appendix, assume that a feasible private-belief mechanism exists. Recall that a jointly feasible threshold measure satisfies \eqref{eq:incentive-moment} and \eqref{eq:mixedfeasibility}; by Lemma~\ref{lem:a-capacity}, the latter is equivalent to \eqref{eq:stoploss}. The first lemma proves attainment and a positive minimum of $R$. The second bounds unsuitable exposure under preservation. The last two construct an improving direction and show that a small step is feasible whenever $y(1-)<1$; together these results establish the strict sacrifice of beneficial use under \eqref{eq:advisor-marginals}.

\begin{lemma}
\label{lem:b-attainment}
The set of attainable pairs $(S,R)$ is compact, and its minimum value of $R$ is strictly positive. The advisor's objective attains a maximum, both without the restriction $S=S^{\mathrm{FB}}$ and with it whenever that restriction is feasible. Every jointly feasible threshold measure $Q$ satisfies
\begin{equation}
 Q((t,1))>0\quad(t\in[0,1)),\qquad
 x_Q(p)<1\quad(p\in(0,1)).
 \label{appimpl:nofiniteschedule}
\end{equation}
\end{lemma}

\begin{proof}
By the direct and symmetry reductions, the attainable set $\mathcal A$ is the image under \eqref{eq:welfareQ} of the nonempty, weakly compact threshold feasible set in Lemma~\ref{lem:a-compactness}. The functions $T_1,vT_0$, extended by zero at one, are continuous. Thus $\mathcal A$ and its intersection with $S=S^{\mathrm{FB}}$ are compact; the latter may be empty. Continuity of $\Phi$ gives attainment, including over asymmetric mechanisms.

If $Q((t_0,1))=0$, suitable agents with beliefs above $t_0$ use the resource almost surely. Any fixed $k+1$ agents meet these conditions with probability $T_1(t_0)^{k+1}>0$, violating capacity. This also gives $x_Q(p)<1$.

If $R=0$, symmetrization gives $\int vT_0\,dQ=0$. The integrand is positive on $(0,1)$ and $Q(\{1\})=0$, so $Q=\delta_0$, contradicting \eqref{appimpl:nofiniteschedule}. Compactness now gives $\min_{\mathcal A}R>0$.
\end{proof}

\begin{lemma}
\label{lem:b-displacement}
For every feasible mechanism--equilibrium pair, let $(x_i,y_i)_{i=1}^n$ be the interims of an outcome-equivalent canonical direct recommendation implementation. With $b$ defined in \eqref{eq:bceiling},
\[
 S^{\mathrm{FB}}-S
 \ge\sum_{i=1}^n\int_0^1(1-p)[y_i(p)-b]_+\,dF(p).
\]
Consequently, $S=S^{\mathrm{FB}}$ implies
\[
 y_i(p)\le b\quad(p\in(0,1)),\qquad y_i(1-)\le b
 \quad(i=1,\ldots,n).
\]
\end{lemma}

\begin{proof}
Use an outcome-equivalent canonical direct implementation. Capacity gives
\[
 \min\left\{k,\sum_i\theta_i\right\}-\sum_i\theta_i a_i
 \ge\sum_i(1-\theta_i)a_i\mathbf1_{\{\sum_{j\ne i}\theta_j\ge k\}}.
\]
The right side is zero if fewer than $k$ agents are suitable; otherwise it counts unsuitable users and is bounded by $k-\sum_i\theta_i a_i$. Independence gives
$\Pr\{\sum_{j\ne i}\theta_j<k\mid p_i=p,\theta_i=0\}=b$ almost everywhere. Hence
\[
 \Pr\{a_i=1,\textstyle\sum_{j\ne i}\theta_j\ge k\mid p_i=p,\theta_i=0\}
 \ge[y_i(p)-b]_+.
\]
Taking expectations proves the bound. If $S=S^{\mathrm{FB}}$, its nonnegative integrands vanish almost everywhere. Monotonicity and full support then give $y_i(p)\le b$ at every interior $p$ and at the upper limit.
\end{proof}

Define the threshold contribution to welfare by
\begin{equation}
 w(\tau)=nu_+[T_1(\tau)-v(\tau)T_0(\tau)]
     =nu_+\int_{\tau}^1\frac{p-\tau}{1-\tau}\,dF(p),\qquad 0\le \tau<1.
 \label{eq:b-welfare-contribution}
\end{equation}
Then $W(Q)=\int w\,dQ$ and $R(Q)=(nu_+/u_-)\int vT_0\,dQ$ by \eqref{eq:welfareQ}.
The function $w$ extends continuously to $[0,1]$ with
$w(1)=0$ and $w(0)=nu_+\mu$.

\begin{lemma}
\label{lem:b-perturbation}
For each jointly feasible $Q$, write $x=x_Q$ and $y=y_Q$. There is $x_{\mathrm c}\in(0,g)$ with $x_{\mathrm c}w(0)>W(Q)$ such that, for every sufficiently large $t<1$, the signed measure
\begin{equation}
 \dot Q_t=x_{\mathrm c}\delta_0+[x(t)-x_{\mathrm c}]\delta_t-Q|_{[0,t]}
 \label{appstruct:truncated-direction}
\end{equation}
has total mass zero and satisfies
\begin{equation}
 \dot J_t:=\int v\,d\dot Q_t>0,\qquad
 \int w\,d\dot Q_t>0,\qquad
 \int vT_0\,d\dot Q_t<0.
 \label{eq:b-direction-signs}
\end{equation}
For $0<\gamma\le1$, $Q+\gamma\dot Q_t$ is a probability measure. Its cumulative functions $x_\gamma(p)=(Q+\gamma\dot Q_t)([0,p])$ and $y_\gamma(p)=(u_+/u_-)\int_{[0,p]}v\,d(Q+\gamma\dot Q_t)$ satisfy
\begin{equation}
 \begin{array}{lll}
 p<t:&x_\gamma(p)=(1-\gamma)x(p)+\gamma x_{\mathrm c},
       &y_\gamma(p)=(1-\gamma)y(p),\\[2pt]
 p\ge t:&x_\gamma(p)=x(p),
       &y_\gamma(p)=y(p)+\gamma\dfrac{u_+}{u_-}\dot J_t.
 \end{array}
 \label{appstruct:truncated-rules}
\end{equation}
\end{lemma}

\begin{proof}
Integration by parts in \eqref{eq:b-welfare-contribution} gives
\begin{equation}
 w(s)-w(t)=nu_+\int_s^t\frac{T_0(\tau)}{(1-\tau)^2}\,d\tau>0
 \qquad(0\le s<t<1).
 \label{eq:b-w-decrease}
\end{equation}
Thus lowering thresholds strictly raises welfare.

Recall $g=S^{\mathrm{FB}}/(n\mu)\in(0,1)$ from \eqref{eq:first-best}.
Lemma~\ref{lem:b-attainment} gives
$W(Q)<u_+S^{\mathrm{FB}}=gw(0)$, so choose
$x_{\mathrm c}\in(0,g)$ such that $x_{\mathrm c}w(0)>W(Q)$.
The pair $(g,0)$, used here only as a capacity benchmark, has stop-loss transform bounded by $L$: for $\ell<g$,
\[
 L(\ell)\ge \frac{H(\mu)}n-\mu\ell=\mu(g-\ell).
\]
For $\ell\ge g$ its transform is zero. Thus $(x_{\mathrm c},0)$ has strict capacity slack:
\begin{equation}
 L(\ell)-\mu(x_{\mathrm c}-\ell)_+>0
 \qquad(0\le\ell<1).
 \label{appstruct:benchmark-slack}
\end{equation}
The gap is at least $\mu(g-x_{\mathrm c})$ below $x_{\mathrm c}$ and equals $L(\ell)>0$ above it, because $H(\eta)/n\sim\eta$ at zero. This comparison concerns capacity only; the truncated threshold perturbation retains the incentive requirement $x(1-)=1$.

For $t<1$ with $x(t)>x_{\mathrm c}$, the direction in \eqref{appstruct:truncated-direction} has odds-moment change
\[
 \dot J_t=v(t)[x(t)-x_{\mathrm c}]-\int_{[0,t]}v\,dQ.
\]
The measure $\dot Q_t$ has total mass zero. Since $x(t)\to1$,
$\int v\,dQ<\infty$, and $w(t)\to0$,
\begin{equation}
 \begin{aligned}
 \dot J_t&\longrightarrow\infty,\\
 \int w\,d\dot Q_t
 &=x_{\mathrm c}w(0)+[x(t)-x_{\mathrm c}]w(t)
     -\int_{[0,t]}w\,dQ
 \longrightarrow x_{\mathrm c}w(0)-W(Q)>0.
 \end{aligned}
 \label{appstruct:direction-limits}
\end{equation}
The same direction reduces harmful use for all sufficiently large $t$:
\begin{equation}
 \begin{aligned}
 \int vT_0\,d\dot Q_t
 &=[x(t)-x_{\mathrm c}]v(t)T_0(t)-\int_{[0,t]}vT_0\,dQ\\
 &\longrightarrow-\int vT_0\,dQ
 =-\frac{u_-}{nu_+}R(Q)<0.
 \end{aligned}
 \label{appstruct:unsuitable-direction}
\end{equation}
Here $v(t)T_0(t)\to0$ because $v(t)T_0(t)\le1-F(t)$, and $R(Q)>0$ by Lemma~\ref{lem:b-attainment}.
For $0<\gamma\le1$, nonnegativity follows from the decomposition
\[
 Q+\gamma \dot Q_t
 =(1-\gamma)Q|_{[0,t]}+Q|_{(t,1)}
  +\gamma x_{\mathrm c}\delta_0
  +\gamma[x(t)-x_{\mathrm c}]\delta_t.
\]
This is a probability measure, and integrating over $[0,p]$ gives \eqref{appstruct:truncated-rules}.
Since $T_1,T_0,vT_0\to0$ at one, the following bound controls the later capacity checks:
\begin{equation}
 E_t:=2T_1(t)+\left[1+2\frac{u_+}{u_-}v(t)\right]T_0(t)
 \longrightarrow0\qquad(t\uparrow1).\qedhere
 \label{appstruct:contact-tail-error}
\end{equation}

\end{proof}

\begin{lemma}
\label{lem:b-endpoint}
If a jointly feasible threshold measure $Q$ satisfies $y_Q(1-)<1$, there is a jointly feasible threshold measure $Q_\gamma$ with
\[
 W(Q_\gamma)>W(Q),\qquad R(Q_\gamma)<R(Q).
\]
\end{lemma}

\begin{proof}
Write $x=x_Q$ and $y=y_Q$, take $x_{\mathrm c}$ from Lemma~\ref{lem:b-perturbation}, and fix
$\ell_*\in(\max\{x_{\mathrm c},y(1-)\},1)$.
Continuity and \eqref{appstruct:benchmark-slack} imply
\[
 \Delta_*:=\min_{0\le\ell\le\ell_*}
       \{L(\ell)-\mu(x_{\mathrm c}-\ell)_+\}>0.
\]
Choose $t$ sufficiently close to one that $x(t)>x_{\mathrm c}$,
$\dot J_t>0$, $\int w\,d\dot Q_t>0$, $\int vT_0\,d\dot Q_t<0$, and $E_t<\Delta_*$.
The preceding limits permit this choice. Choose
\[
 0<\gamma<\min\left\{1,
 \frac{\ell_*-y(1-)}{(u_+/u_-)\dot J_t}\right\},
 \qquad Q_\gamma=Q+\gamma \dot Q_t.
\]
The measure $Q_\gamma$ is a probability measure and satisfies
\[
 \frac{u_+}{u_-}\int v\,dQ_\gamma
 =y(1-)+\gamma\frac{u_+}{u_-}\dot J_t<\ell_*<1.
\]
Thus $Q_\gamma$ satisfies \eqref{eq:incentive-moment}.

For capacity, compare its rules with the mixture
$((1-\gamma)x+\gamma x_{\mathrm c},(1-\gamma)y)$.
By \eqref{appstruct:truncated-rules}, the two pairs coincide below $t$.
For $p\ge t$, using $0<\dot J_t\le v(t)$ gives
\[
 \begin{aligned}
 0\le x_\gamma(p)-[(1-\gamma)x(p)+\gamma x_{\mathrm c}]
     &=\gamma[x(p)-x_{\mathrm c}]\le\gamma,\\
 0\le y_\gamma(p)-(1-\gamma)y(p)
     &=\gamma\left[y(p)+\frac{u_+}{u_-}\dot J_t\right]
       \le\gamma\left[y(1-)+\frac{u_+}{u_-}v(t)\right].
 \end{aligned}
\]
The integrated differences are at most $\gamma E_t$. Convexity and the one-Lipschitz property of the positive part give, for $0\le\ell\le\ell_*$,
\[
 \begin{aligned}
 L_{Q_\gamma}(\ell)
 &\le(1-\gamma)L_Q(\ell)
      +\gamma\mu(x_{\mathrm c}-\ell)_++\gamma E_t\\
 &\le L(\ell)-\gamma(\Delta_*-E_t)<L(\ell).
 \end{aligned}
\]
Above $\ell_*$, unsuitable positive parts vanish. Suitable positive parts are unchanged above $t$ and cannot increase below it, where $x_\gamma$ lies between $x$ and $x_{\mathrm c}<\ell_*$. Thus $L_{Q_\gamma}\le L_Q\le L$ there as well. Lemma~\ref{lem:a-capacity} gives feasibility, while
\[
 \begin{aligned}
 W(Q_\gamma)-W(Q)&=\gamma\int w\,d\dot Q_t>0,\\
 R(Q_\gamma)-R(Q)&=\gamma\frac{nu_+}{u_-}\int vT_0\,d\dot Q_t<0.
 \end{aligned}
\]
\end{proof}

\begin{proof}[Proof of Theorem~\ref{thm:distortion}]
Lemma~\ref{lem:b-attainment} proves compactness of the attainable set and Theorem~\ref{thm:distortion}(i). Continuity of $\Phi$ therefore ensures attainment. For part (ii), first consider a symmetric optimum represented by $Q$. Set
\[
 A=\Phi_S(S(Q),R(Q))\ge0,\qquad
 B=-\Phi_R(S(Q),R(Q))>0.
\]
If $y_Q(1-)<1$, then $R(Q)>0$ and $S(Q)\le n\mu g$ imply $AS(Q)-BR(Q)<An\mu g$. Choose $x_{\mathrm c}\in(0,g)$ sufficiently close to $g$ that both $x_{\mathrm c}w(0)>W(Q)$ and $An\mu x_{\mathrm c}>AS(Q)-BR(Q)$ hold; for $A=0$, the latter follows from $R(Q)>0$. The perturbation in Lemma~\ref{lem:b-perturbation} then satisfies
\[
 \int\!\left[AnT_1-B\frac{nu_+}{u_-}vT_0\right]d\dot Q_t
 \longrightarrow An\mu x_{\mathrm c}-[AS(Q)-BR(Q)]>0.
\]
Choose $t$ additionally large enough for this expression to be positive. The capacity and odds-bound argument in Lemma~\ref{lem:b-endpoint} makes every sufficiently small positive step $Q_\gamma=Q+\gamma\dot Q_t$ feasible. Write $\dot S_t=n\int T_1\,d\dot Q_t$ and $\dot R_t=(nu_+/u_-)\int vT_0\,d\dot Q_t$. Differentiability gives
\[
 \Phi(S(Q_\gamma),R(Q_\gamma))-\Phi(S(Q),R(Q))
 =\gamma(A\dot S_t-B\dot R_t)+o(\gamma)>0
\]
for small enough $\gamma>0$, contradicting optimality. Thus every symmetric optimum has $y_Q(1-)=1$.

For an asymmetric optimum, the canonical reduction and Lemma~\ref{lem:a-symmetry} preserve $S$, $R$, and feasibility and give $\bar y=n^{-1}\sum_i y_i$. The symmetrized mechanism is therefore optimal for the same objective, so
\[
 1=\bar y(1-)=\frac1n\sum_i y_i(1-),\qquad 0\le y_i(1-)\le1.
\]
All summands must equal one. Because $b<1$, Lemma~\ref{lem:b-displacement} gives $S<S^{\mathrm{FB}}$, proving part (ii).

Part (iii) follows from Lemma~\ref{lem:observed-beliefs}: simulate the private-belief outcome using observed beliefs and then complete its suitable selections.\par\medskip

To prove the local implication \eqref{eq:local-improvement}, take a feasible pair $(S^{\mathrm{FB}},R)$ and fix $M,\delta>0$. Set $M_0=\max\{M,u_+/u_-\}$, and let $(S_1,R_1)$ maximize the auxiliary objective $M_0S-R$. This objective satisfies \eqref{eq:advisor-marginals}, so part (ii) gives $S_1<S^{\mathrm{FB}}$. Moreover,
\[
 M_0S_1-R_1>M_0S^{\mathrm{FB}}-R,
\]
since otherwise the preserving pair would also maximize this objective. Consequently,
\[
 R-R_1>M_0(S^{\mathrm{FB}}-S_1)>0.
\]
A mixture of the two direct recommendation kernels is feasible: every joint incentive inequality is linear in the kernel, and each component respects capacity almost surely. Choose
\[
 0<\lambda<\min\left\{1,
 \frac{\delta}{S^{\mathrm{FB}}-S_1},
 \frac{\delta}{R-R_1}\right\}.
\]
Assigning probability $\lambda$ to the optimizer gives use levels
$S'=(1-\lambda)S^{\mathrm{FB}}+\lambda S_1$ and
$R'=(1-\lambda)R+\lambda R_1$. They satisfy
\[
 \begin{gathered}
 0<S^{\mathrm{FB}}-S'<\delta,\qquad 0<R-R'<\delta,\\
 R-R'>M(S^{\mathrm{FB}}-S'),\qquad
 u_+S'-u_-R'>u_+S^{\mathrm{FB}}-u_-R.
 \end{gathered}
\]
The last inequality uses $M_0\ge u_+/u_-$. This proves the claimed local improvement, with both use changes arbitrarily small.
\end{proof}

\section{Proof and construction for Example~\ref{ex:uniform}}
\label{app:uniform-example}

\subsection{Construction and feasibility of the optimal uniform example}
\label{app:uniform-optimal-primal}

The interims in Example~\ref{ex:uniform} are defined by a contraction
and a scalar equation.
Throughout, $n=3$, $k=2$, $F$ is uniform, and $u_+=u_-=1$.
Write
\[
 Y_0(p)=1-\left[\frac12+\frac{(1-p)^2}{2}\right]^2,
 \qquad q_0(p)=\frac{(1-p)^2[1+(1-p)^2]}p.
\]

\paragraph{The universal upper tail.}
There is a unique absolutely continuous function $f$ on $[0,1/8]$
such that $f(0)=0$, $s\le f'(s)\le3s$ almost everywhere, and, on
putting
\begin{equation}
 h(s)=(s^{-1}-1)f(s)+\int_0^s\frac{f(t)}{t^2}\,dt,
 \qquad h(0)=0,
 \label{eq:uniform-tail-h}
\end{equation}
it satisfies
\begin{equation}
 f(s)=\left[s-\frac{s^2}{2}
       +\frac{\{h^{-1}(f(s))\}^2}{2}\right]^2.
 \label{eq:uniform-tail-fixedpoint}
\end{equation}
Here $h$ is strictly increasing, with
$1-s\le h'(s)\le3(1-s)$ almost everywhere. The inverse in
\eqref{eq:uniform-tail-fixedpoint} is consequently well defined.

To prove existence and uniqueness, equip the displayed class of
absolutely continuous functions with the complete metric
$\|f_1-f_2\|:=\sup_{0<s\le1/8}|f_1(s)-f_2(s)|/s^2$.
For a proposed $f$, let $(\mathcal Tf)(s)$ be the unique zero in
$[(s-s^2/2)^2,(9/8)s^2]$ of
\[
 \Psi_f(z;s)=\sqrt z-s+s^2/2-\tfrac12\{h^{-1}(z)\}^2.
\]
On this interval, $\delta=h^{-1}(z)\le6s^2/5$,
$\partial_z\Psi_f\ge1/(3s)$ almost everywhere, and $\Psi_f$ has opposite
signs at the two endpoints. Implicit differentiation, with $\delta=h^{-1}((\mathcal Tf)(s))$, gives
\[
 (\mathcal Tf)'(s)=
 \frac{2\sqrt{\mathcal Tf(s)}(1-s)}
 {1-2\delta\sqrt{\mathcal Tf(s)}/h'(\delta)}\in[s,3s]
 \quad\text{almost everywhere}.
\]
Thus $\mathcal T$ maps the class into itself. If two proposed functions
are distance $D$ apart, their $h$ functions differ at $t$ by at most
$2tD$. The inverse Lipschitz bound $8/7$ and
$\delta\le6s^2/5$ imply
$|\Psi_{f_1}(z;s)-\Psi_{f_2}(z;s)|\le4s^4D$. Therefore
\[
 \|\mathcal Tf_1-\mathcal Tf_2\|
 \le12(1/8)^3\|f_1-f_2\|=\frac3{128}\|f_1-f_2\|.
\]
The contraction theorem gives a unique fixed point of $\mathcal T$.
Any solution of \eqref{eq:uniform-tail-fixedpoint} in the stated class
is such a fixed point: $f(s)\le3s^2/2$ and
$h(t)\ge t-t^2/2\ge15t/16$ give
$\delta=h^{-1}(f(s))\le8s^2/5$, and hence
\[
 (s-s^2/2)^2\le f(s)
 \le s^2\left(1+\frac{32}{25}s^3\right)^2<\frac98s^2
 \qquad(0<s\le1/8).
\]
Thus it lies in the root interval defining $\mathcal T$, proving uniqueness
in the whole stated class. The contraction also supplies a direct
way to evaluate the tail. In particular,
\begin{equation}
 \begin{gathered}
 f_0(s):=(s-s^2/2)^2\le f(s)\le f_0(s)+\tfrac32s^5,\\
 h_0(s):=2s-\tfrac52s^2+\tfrac43s^3-\tfrac14s^4
       \le h(s)\le h_0(s)+2s^4.
 \end{gathered}
 \label{eq:uniform-tail-bounds}
\end{equation}
The first error bound follows by expanding
$(s-s^2/2+\delta^2/2)^2$ and using $\delta\le6s^2/5$;
the second follows from \eqref{eq:uniform-tail-h}.

For $p\ge7/8$, set
\[
 X(p)=1-f(1-p),\qquad Y(p)=1-h(1-p),\qquad
 \phi(p)=1-h^{-1}(f(1-p)).
\]
These functions are locally absolutely continuous. Equations
\eqref{eq:uniform-tail-h}--\eqref{eq:uniform-tail-fixedpoint} imply
\begin{equation}
 \begin{gathered}
 dY=v\,dX,\qquad X(1-)=Y(1-)=1,\\
 X(p)=Y(\phi(p))
 =1-[T_1(p)+T_0(\phi(p))]^2,\\
 0<1-\phi(p)\le\tfrac65(1-p)^2<1-p,
 \qquad \phi'(p)\ge(1-p)/3\quad\text{almost everywhere}.
 \end{gathered}
 \label{eq:uniform-tail-rank}
\end{equation}
All subsequent derivative statements may likewise be read almost
everywhere; no additional smoothness is needed.

\paragraph{The thresholds and the two jumps.}
For a proposed $\alpha\in[.889,.890]$, put
$\beta=(1+\alpha)/2$ and $e=1-\beta$. Define
\begin{align*}
 d_\alpha&=\frac{1-\alpha}{\alpha}Y_0(\alpha),&
 d_\beta&=\frac{1-\beta}{\beta}[Y(\beta)-Y_0(\beta)],\\
 x_{\beta-}&=X(\beta)-d_\beta,&
 x_{\ast}&=x_{\beta-}-d_\alpha-\int_\alpha^\beta q_0(t)\,dt.
\end{align*}
The candidate interims are
\begin{equation}
 (x(p),y(p))=
 \begin{cases}
 (x_{\ast},0),&p<\alpha,\\
 \left(x_{\ast}+d_\alpha+\displaystyle\int_\alpha^p q_0(t)\,dt,
       Y_0(p)\right),&\alpha\le p<\beta,\\
 (X(p),Y(p)),&\beta\le p<1.
 \end{cases}
 \label{eq:uniform-optimal-interims}
\end{equation}
Thus $x_{\beta-}=x(\beta-)$.
The threshold $\alpha$ is the unique root in $(.889,.890)$ of
\begin{equation}
 \int_0^1p x(p)\,dp+\int_\beta^1(1-p)y(p)\,dp
 =\eta-\frac{\eta^3}{3},\qquad
 \eta=\frac12+T_0(\beta).
 \label{eq:uniform-optimal-root}
\end{equation}
Thus the root imposes one capacity equality, and all quantities in
it are already given by the contraction and explicit integrals.

Here are certified existence and sign checks for this definition.
Let $d(s)=f(s)-f_0(s)$ and write the left side of
\eqref{eq:uniform-optimal-root} minus its right side as
$\mathcal B(\alpha)$. Direct integration gives
\begin{align*}
 \mathcal B(\alpha)&=\mathcal P(e)+D(e),\\
 \mathcal P(e)&=\frac1{24}-\frac78e+\frac94e^2-\frac76e^3
       -\frac13e^4+\frac{13}{60}e^5-\frac52e^6,\\
 D(e)&=\frac{\alpha e}{2}\int_0^e\frac{d(s)}{s^2}\,ds
       -\int_0^e(3/2-2s)d(s)\,ds.
\end{align*}
By \eqref{eq:uniform-tail-bounds},
$0\le D(e)\le3\alpha e^5/16$ on the proposed interval.
Rational evaluation of $\mathcal P$ consequently gives
\[
 \mathcal B(.889)<-.0001677508,
 \qquad \mathcal B(.890)>.0001508021.
\]
Continuity proves existence. Differentiating the explicit
construction gives
\[
 \mathcal B'(\alpha)=\frac{2\alpha-1}{4}
   \{2Y_0(\alpha)+Y(\beta)-Y_0(\beta)\}>0,
\]
which proves uniqueness. Alternatively, the above polynomial has
$\mathcal P'(e)<-.636$ and the error term has $0\le D'(e)<.000009$
throughout $e\in[.055,.0555]$.

For every $\alpha$ in the proposed interval, the same tail bounds
and elementary integration give
\[
 \begin{gathered}
 .8950078<x_{\ast}<.8961717,
 \qquad .8964560<Y(\beta)<.8973430,\\
 .9883388<x_{\beta-}<.9885271,
 \qquad .9970875<X(\beta)<.9971391,\\
 .0919305<d_\alpha<.0928846,
 \qquad .0086120<d_\beta<.0087487.
 \end{gathered}
\]
For example, $q_0$ is decreasing on $[.889,.945]$; upper and lower
Riemann sums on $[.889,.945]$ and $[.890,.9445]$, each with
100 equal subintervals, give
$.0004249<\int_\alpha^\beta q_0<.0004464$, sufficient for all
these strict bounds. They imply, writing
$x_\alpha:=x(\alpha)=x_{\ast}+d_\alpha$,
\begin{equation}
 \begin{gathered}
 0<Y_0(\alpha)<Y_0(\beta)<x_{\ast}<Y(\beta)
   <x_\alpha<x_{\beta-}<X(\beta)<1,\\
 [T_1(\alpha)+T_0(\beta)]^2<1-x_{\beta-}.
 \end{gathered}
 \label{eq:uniform-branch-order}
\end{equation}
For the second inequality,
$T_1(\alpha)+T_0(\beta)\le.106379625$, and the difference between
$1-x_{\beta-}$ and the square of this bound exceeds $.0001562$.

Write $\Delta x(t)=x(t)-x(t-)$ and $\Delta y(t)=y(t)-y(t-)$ for the jumps at $t$. Both jumps are positive and satisfy $\Delta y=v\Delta x$.
The continuous branches do too, since $Y_0'=vq_0$ and
\eqref{eq:uniform-tail-rank} holds. Thus
\[
 Q^*=x_{\ast}\delta_0+d_\alpha\delta_\alpha
   +q_0(p)\mathbf1_{(\alpha,\beta)}\,dp
   +d_\beta\delta_\beta
   +f'(1-p)\mathbf1_{(\beta,1)}\,dp
\]
is a probability measure satisfying $\int v\,dQ^*=1$ and producing
\eqref{eq:uniform-optimal-interims}. Lemma~\ref{lem:a-ic} therefore
gives all joint reporting--action incentives.

\paragraph{All capacity inequalities.}
For $n=3,k=2$, the capacity stop-loss transform is
$L(\ell)=\frac23(1-\ell)^{3/2}$. Write $\Delta=L-L_{Q^*}$.
For every $\ell\ge X(\beta)$,
\eqref{eq:uniform-tail-rank} makes the mass of the interim upper-level
set exactly $\sqrt{1-\ell}$. Hence $L_{Q^*}'=L'$ there, and both vanish
at one, proving equality of the transforms on this interval.

For $Y_0(\alpha)\le\ell\le Y_0(\beta)$, the upper-level set
contains every suitable type and the unsuitable types above the
unique $t$ satisfying $Y_0(t)=\ell$. Its mass is again
$\frac12+T_0(t)=\sqrt{1-\ell}$. Equation
\eqref{eq:uniform-optimal-root} fixes the integration constant,
so $L_{Q^*}=L$ throughout this second interval. Below $Y_0(\alpha)$,
and between $Y_0(\beta)$ and $x_{\ast}$, the ordering
\eqref{eq:uniform-branch-order} makes $L_{Q^*}$ the tangent line to
the convex function $L$ at the adjacent equality endpoint.
Thus $\Delta\ge0$ on both intervals.

It remains to check $x_{\ast}\le\ell\le X(\beta)$.
For $x_{\ast}\le\ell\le x_{\beta-}$, the mass $\eta_\ell$ of the upper-level
set is at most $T_1(\alpha)+T_0(\beta)$. Therefore
\[
 \Delta'(\ell)=\eta_\ell-\sqrt{1-\ell}<0
\]
by \eqref{eq:uniform-branch-order}. For
$x_{\beta-}\le\ell\le X(\beta)$, the suitable cutoff is fixed at
$\beta$ and the unsuitable cutoff $q$ satisfies $Y(q)=\ell$.
Thus $\eta_\ell=T_1(\beta)+T_0(q)$, with
$\eta_{X(\beta)}=\sqrt{1-X(\beta)}$. Moreover,
\[
 \frac{d}{d\ell}[\eta_\ell-\sqrt{1-\ell}]
 =-\frac{1-q}{Y'(q)}+\frac1{2\sqrt{1-\ell}}>0
 \quad\text{almost everywhere}:
\]
indeed $Y'(q)=h'(1-q)\ge q\ge\beta>.944$.
Consequently $\Delta'\le0$ on this last interval too. Since
$\Delta(X(\beta))=0$, $\Delta\ge0$ throughout the remaining range.
All stop-loss inequalities hold. Lemma~\ref{lem:a-capacity}
therefore supplies an implementing direct recommendation kernel
with pointwise capacity. The capacity equalities just proved are
exactly the two families used by the optimality certificate below.

\subsection{Global optimality}
\label{app:uniform-optimal-dual}
We give a dual certificate for the measure $Q^*$ constructed above.
Put $\beta=(1+\alpha)/2$, $d_0=1-\beta<1/16$,
$M=2/(2\alpha-1)=2/(1-4d_0)$, and $A=Md_0^2$.
Recall that $Q^*$ has support in $\{0\}\cup[\alpha,1)$, has
$\int v\,dQ^*=1$, and binds the capacity inequalities with cutoffs
$(0,t)$ for $\alpha\le t\le\beta$ and $(s,\phi(s))$ for
$\beta\le s<1$. The tail construction gives
\begin{equation}
 0<1-\phi(s)\le\frac65(1-s)^2,
 \qquad \phi'(s)\ge\frac{1-s}{3}\quad\text{a.e.}
 \label{exopt:dual-tailbounds}
\end{equation}
The function $\phi$ is strictly increasing and locally absolutely
continuous; all differential identities below are understood almost
everywhere.

Recall $H(\eta)/3=\eta-\eta^3/3$. For cutoffs $s,t\in[0,1]$ and a threshold $z\in[0,1)$, define
\[
 C_{s,t}(z)=T_1(\max\{s,z\})+
             v(z)T_0(\max\{t,z\}).
\]
The capacity inequality divided by three is
$\int C_{s,t}\,dQ\le H(T_1(s)+T_0(t))/3$.
For a nonnegative measure $\nu$ on cutoff pairs and $\lambda\ge0$, consider
\begin{equation}
 G(z)=z+\lambda v(z)+
  \int\bigl[C_{s,t}(z)-C_{s,t}(0)\bigr]\,\nu(ds,dt).
 \label{exopt:dual-function}
\end{equation}
If $G\ge0$, every feasible $Q$ satisfies
\begin{equation}
 \int z\,Q(dz)\ge-\lambda-
 \int\bigl[H(T_1(s)+T_0(t))/3-C_{s,t}(0)\bigr]\,\nu(ds,dt).
 \label{exopt:dual-bound}
\end{equation}
We construct multipliers making this bound an equality at $Q^*$.
The subtraction of $C_{s,t}(0)$ keeps all integrals finite despite the
infinite total mass of the tail multiplier.

Place density
\[
 \frac{M(1-\alpha)}{(1-t)^2}\,dt,
 \qquad \alpha<t<\beta,
\]
on the cutoff pairs $(0,t)$. This first part of $\nu$ has total mass $M$
and $T_0$-weighted mass $A$; its cumulative mass below $t\in(\alpha,\beta)$ is
$N(t)=M(t-\alpha)/(1-t)$.
Place a density $m(s)\,ds$, constructed next, on the pairs $(s,\phi(s))$,
and choose $\lambda$ so that
\begin{equation}
 \lambda+\int_\beta^1T_0(\phi(s))m(s)\,ds=A.
 \label{exopt:dual-moment}
\end{equation}
Direct substitution into \eqref{exopt:dual-function} then gives
\begin{equation}
 G(z)=\frac{z(\alpha-z)^2}{(2\alpha-1)(1-z)}
 \quad(0\le z\le\alpha).
 \label{exopt:dual-gap}
\end{equation}
For $\alpha<z<\beta$, differentiating and substituting $N(z)$ gives
$G'(z)=0$. Hence $G=0$ on $[\alpha,\beta]$.

To construct the tail density, write $\psi=\phi^{-1}$ and extend
expressions involving $\psi(t)$ by zero for $t<\phi(\beta)$. Define
\begin{align*}
 P_0(t)&=\frac{A}{t(1-t)^2}-\frac{3M-2}{2t}
                       +\frac{2M(1-t)}{t},\\
 m_0(t)&=P_0'(t)
 =\frac{1}{t^2}\left\{\frac{A(3t-1)}{(1-t)^3}
                         -\frac{M+2}{2}\right\}.
\end{align*}
We have $P_0(\beta)=0$ and
\begin{equation}
 \frac{7A}{4(1-t)^3}\le m_0(t)\le
 \frac{3A}{(1-t)^3},\qquad \beta\le t<1.
 \label{exopt:dual-reference-bounds}
\end{equation}
Indeed, for the lower bound, multiplying by $(1-t)^3/A$ and using
$(M+2)/(2M)\le1$ gives a lower bound
$2-4d_0\ge7/4$; the upper bound follows from
$(3t-1)/t^2\le3$.

For a nonnegative density $m$, put
\[
 P(t)=\int_\beta^t m(s)\,ds,\qquad
 P_\psi(t)=P(\psi(t)),\qquad
 J(t)=\int_{\phi(\beta)}^t(1-s)P_\psi(s)\,ds,
\]
where $P_\psi=J=0$ before $\phi(\beta)$. Define the positive linear operator
\begin{equation}
 (Km)(t)=\frac{m(\psi(t))}{\phi'(\psi(t))}
       +\frac{P_\psi(t)}{t(1-t)}
       +\frac{(3t-1)J(t)}{t^2(1-t)^3}.
 \label{exopt:dual-operator}
\end{equation}
We solve $m=m_0-Km$ in the weighted supremum norm
$\|m\|=\operatorname*{ess\,sup}_{\beta\le t<1}|m(t)|/m_0(t)$.
Here are explicit bounds verifying that $K$ is a contraction.
For $t<\phi(\beta)$, $Km=0$. For $t\ge\phi(\beta)$, set
$\delta=1-t$, $d=1-\psi(t)$, and
$\delta_{\max}=1-\phi(\beta)\le6d_0^2/5$.
For the input $m_0$, use
$P_0(u)\le A/[u(1-u)^2]$, \eqref{exopt:dual-tailbounds}, and
\eqref{exopt:dual-reference-bounds}. The three terms in
\eqref{exopt:dual-operator}, divided by $m_0(t)$, are respectively at most
\[
 \frac{1296}{175}\delta,
 \qquad\frac{24}{35\beta^2}\delta,
 \qquad\frac{432}{175\beta}d_0^2.
\]
For the last bound one uses
$J(t)\le36Ad_0^2/(25\beta)$.
Since $\delta\le\delta_{\max}$, $d_0\le1/16$, and $\beta\ge15/16$,
their sum is less than $0.051<1/8$.
Thus $\|K\|<1/8$, and the contraction $m\mapsto m_0-Km$
maps the order interval $[7m_0/8,m_0]$ into itself. Its fixed point
satisfies $7m_0/8\le m\le m_0$, so it defines a nonnegative,
locally finite tail multiplier.

Integrating its defining equation gives
\[
 P(t)=P_0(t)-P_\psi(t)-\frac{J(t)}{t(1-t)^2}.
\]
Equivalently,
\begin{equation}
 t(1-t)(P+P_\psi)'=(3t-1)P+(3t-2)P_\psi+M(6t-3)-2.
 \label{exopt:dual-tail-equation}
\end{equation}
The choice \eqref{exopt:dual-moment} has $\lambda>0$, because
\[
 \int_\beta^1T_0(\phi(s))m(s)\,ds
 \le\int_\beta^1\frac{18}{25}(1-s)^4
                   \frac{3A}{(1-s)^3}\,ds
 =\frac{27}{25}Ad_0^2<A.
\]
For $t\ge\beta$, set
\[
 \Lambda(t)=\lambda+\int_\beta^1T_0(\max\{\phi(s),t\})m(s)\,ds.
\]
Then $\Lambda(\beta)=A$, $\Lambda'=-(1-t)P_\psi$, and
\[
 G'(t)=1+M(1/2-2t)-t(P+P_\psi)+\frac{\Lambda(t)}{(1-t)^2}.
\]
Equation~\eqref{exopt:dual-tail-equation} and the initial value imply
\[
 \Lambda(t)=(1-t)^2\{t(P+P_\psi)-1-M/2+2Mt\}.
\]
Indeed, differentiating verifies the identity, and it holds at $\beta$
by $A=Md_0^2$ and $M=2/(1-4d_0)$.
Consequently $G'=0$ on $[\beta,1)$, proving that $G=0$ on
$[\alpha,1)$ and $G\ge0$ everywhere.

For completeness, the normalized integrals above are absolutely
convergent. The tail density is $O((1-s)^{-3})$. For $z\ge\beta$, the negative suitable part is bounded by
$3A\int_\beta^z(1-s)^{-2}\,ds\le3A/(1-z)$; it is zero for $z<\beta$. The positive unsuitable
part is bounded by
$v(z)\int_\beta^1T_0(\phi(s))m(s)\,ds\le Av(z)$.
Thus $\int|C_{s,t}(z)-C_{s,t}(0)|\,d\nu=O(1+v(z))$,
which is integrable against every feasible $Q$ because $\int v\,dQ\le1$.
Also
$H(T_1(s)+T_0(\phi(s)))/3-T_1(s)=O((1-s)^3)$.
Thus one can first integrate over truncated cutoff intervals and then
pass to the limit in \eqref{exopt:dual-bound}.

The support of $Q^*$ lies in the zero set of $G$, its odds moment equals
one, and each capacity inequality receiving a multiplier is an equality.
It therefore attains the bound \eqref{exopt:dual-bound}, globally minimizing
$\int z\,dQ$ and maximizing
$W=(3/2)(1-\int z\,dQ)$. The direct reduction and symmetrization extend
this bound to all feasible mechanisms, including asymmetric mechanisms.

\subsection{Minimum-harm preservation}
\label{app:uniform-preserving}

Keep $n=3$, $k=2$, uniform beliefs, and $u_+=u_-=1$. We construct
a policy minimizing $R$ among all feasible policies with
$S=S^{\mathrm{FB}}$. Let $\tau\in(0,1)$ be the unique solution of
\[
 \tau^3+2\tau^2=2,\qquad \tau\simeq0.8392867552,
\]
and put $x_0=3/4+\tau^2/4-\tau^4/12$. The superscript $\mathrm P$
denotes the preserving policy. Its interims are
\begin{equation}
 \begin{aligned}
 x^{\mathrm P}(p)&=
 \begin{cases}
 x_0,&p<\tau,\\
 1-(1-p^2)^2/4,&p\ge\tau,
 \end{cases}
 &
 y^{\mathrm P}(p)&=
 \begin{cases}
 0,&p<\tau,\\
 p^3/3+p^4/4+1/6,&p\ge\tau.
 \end{cases}
 \end{aligned}
 \label{expres:interims}
\end{equation}

To implement suitable recommendations, give reports at least $\tau$
priority in decreasing report order, followed by a pool of reports
below $\tau$; break all ties uniformly and select up to two suitable
agents. A suitable report $p\ge\tau$ is excluded only if both opponents
are suitable and report more than $p$, yielding the upper branch of
$x^{\mathrm P}$. A suitable report below $\tau$ is excluded with probability
\[
 \frac{(1-\tau^2)^2}{4}
 +\frac{\tau^2(1-\tau^2)}{4}
 +\frac{\tau^4}{12}
 =1-x_0.
\]
The three terms cover two, one, and zero opponents above the cutoff,
respectively, when both opponents are suitable.

Provisionally fill spare places with unsuitable agents in decreasing
report order, breaking ties uniformly. Each opponent has higher
priority than an unsuitable report $p$ with probability
$1/2+T_0(p)=1-p+p^2/2$, so its provisional selection probability is
\[
 z(p)=1-(1-p+p^2/2)^2>0.
\]
Retain each provisional unsuitable selection independently with
probability $y^{\mathrm P}(p)/z(p)$. Below $\tau$ this probability is zero;
for $p\ge\tau>1/2$ it is at most one because
\[
 z(p)-y^{\mathrm P}(p)
 =\frac{(1-p)(3p^3-p^2+11p-1)}6\ge0.
\]
Disclose only the final use or nonuse recommendations.
This procedure respects capacity at every report--state profile and
implements \eqref{expres:interims}.

Between jumps, $dy^{\mathrm P}=v\,dx^{\mathrm P}$. At the cutoff,
\[
 \Delta x^{\mathrm P}(\tau)=\frac{\tau^2}{4}-\frac{\tau^4}{6},\qquad
 v(\tau)\Delta x^{\mathrm P}(\tau)
 =\frac{\tau^3}{3}+\frac{\tau^4}{4}+\frac16
 =\Delta y^{\mathrm P}(\tau),
\]
where the middle equality follows from $\tau^3+2\tau^2=2$.
Thus the probability measure
\[
 Q^{\mathrm P}=x_0\delta_0+
 \left(\frac{\tau^2}{4}-\frac{\tau^4}{6}\right)\delta_\tau
 +t(1-t^2)\mathbf1_{\{t>\tau\}}\,dt
\]
represents these interims. Its odds moment is
$y^{\mathrm P}(1-)=3/4<1$, so Theorem~\ref{thm:implementation}(i) establishes
truthful obedience against every joint reporting--action deviation.
Suitable agents always have priority over unsuitable agents, giving
\begin{equation}
 \begin{aligned}
 S^{\mathrm P}&=S^{\mathrm{FB}}=\frac{11}{8},\\
 R^{\mathrm P}&=\frac3{40}-\frac{\tau^3}{8}+\frac{\tau^5}{20}
       \simeq0.02192253608,\\
 W^{\mathrm P}&=\frac{11}{8}-R^{\mathrm P}\simeq1.35307746392.
 \end{aligned}
 \label{expres:outcome}
\end{equation}
The expression for $R^{\mathrm P}$ follows from
$R=(3/2)\int(t-t^2)\,dQ$ and the displayed $Q^{\mathrm P}$.

We next prove optimality among all preserving mechanisms. By canonical
reduction and symmetrization, it suffices to consider a symmetric
canonical pair $(x,y)$. Write
\[
 M=y(1-)=\int_0^1\frac{1-x(p)}{(1-p)^2}\,dp\le\frac34,
 \qquad M^{\mathrm P}=\frac34.
\]
The identity is \eqref{appimpl:tailcharge}; the inequality follows from
Lemma~\ref{lem:b-displacement}, since here $b=3/4$.
Uniform beliefs also give
\[
 R=\frac32\int_0^1(2p-1)x(p)\,dp.
\]
Define the positive multiplier and cost function
\[
 \lambda=\frac{3(1-\tau)^2}{2(2\tau-1)},\qquad
 c(p)=3-\frac{3}{2p}-\frac{\lambda}{p(1-p)^2},
\]
and let $\bar c(p)=c(\tau)$ for $p<\tau$ and
$\bar c(p)=c(p)$ for $p\ge\tau$. Direct calculation gives
\[
 c(\tau)=-\frac{6(1-\tau)}{2\tau-1}<0,\qquad
 c'(p)=\frac{3}{2p^2}
 -\frac{\lambda(3p-1)}{p^2(1-p)^3}<0
 \quad(p\ge\tau).
\]
For the last inequality, $(1-p)^3/(3p-1)$ decreases on $(1/3,1)$,
and its value at $\tau$ is strictly below $2\lambda/3$.
Hence the exclusion weight $-\bar c(p)$ is positive,
constant below $\tau$, and strictly increasing above it.

To compare $c$ with $\bar c$, define, for $0\le t\le\tau$,
\[
 J(t)=\int_0^t p[c(p)-c(\tau)]\,dp
 =-\frac{3t(\tau-t)^2}{2(2\tau-1)(1-t)}\le0.
\]
Since $J(0)=J(\tau)=0$ and $x$ is nondecreasing, integration by parts
gives
\[
 \int_0^\tau p[c(p)-\bar c(p)]x(p)\,dp
 =-\int_{(0,\tau)}J(p)\,dx(p)\ge0.
\]
Equality holds for $x^{\mathrm P}$, which is constant below $\tau$.
Consequently,
\begin{align}
 R-R^{\mathrm P}+\lambda(M-M^{\mathrm P})
 &=\int_0^1 pc(p)[x(p)-x^{\mathrm P}(p)]\,dp\notag\\
 &\ge\int_0^1 p\bar c(p)[x(p)-x^{\mathrm P}(p)]\,dp\notag\\
 &=-\int_0^1 p\bar c(p)
       \big[(1-x(p))-(1-x^{\mathrm P}(p))\big]\,dp\ge0.
 \label{expres:optimality}
\end{align}
The final inequality follows directly from preservation. Every
preserving implementation excludes no suitable agent when at most two
are suitable, and excludes exactly one when all three are suitable,
almost surely. The proposed policy then excludes an agent with the
smallest $-\bar c(p)$: reports below $\tau$ are tied, and higher reports
have strictly greater weights. It therefore minimizes weighted suitable
exclusion at each profile. Taking expectations yields the last
inequality in \eqref{expres:optimality}.

All integrals in this comparison are finite. Near one,
$-p\bar c(p)$ is bounded by a constant times $1+(1-p)^{-2}$, and
the integrals of $(1-x)/(1-p)^2$ and $(1-x^{\mathrm P})/(1-p)^2$ are
$M$ and $M^{\mathrm P}$. The difference involving $c$ is likewise absolutely
integrable because
$|x-x^{\mathrm P}|\le(1-x)+(1-x^{\mathrm P})$.
Finally, $M\le M^{\mathrm P}$ and \eqref{expres:optimality} imply
\[
 R-R^{\mathrm P}\ge\lambda(M^{\mathrm P}-M)\ge0.
\]
This proves that \eqref{expres:outcome} gives the smallest possible
harmful use subject to preservation, including among
asymmetric mechanisms.

\section{Proofs of the limits of coordination}
\label{app:coordination}

\subsection{Private-belief claims in Proposition~\ref{prop:feasible-costs}}
\label{app:private-feasibility}

\begin{proof}[Proof of Proposition~\ref{prop:feasible-costs}(i)--(ii)]
We prove parts (i) and (ii), the integrability criterion, and the sufficient preservation bound; the observed-belief claims are proved in the next subsection.
Common scaling of $(u_+,u_-)$ scales every utility difference and leaves capacity unchanged, so only their ratio matters.
To verify the expected-minimum-odds interpretation of $I_k(F)$, conditional independence gives
\[
 \Pr\!\left\{\min_{1\le j\le k}p_j>t\,\middle|\,
 \theta_1=\cdots=\theta_k=1\right\}
 =\left(\frac{T_1(t)}{\mu}\right)^k.
\]
Since $v(p)=\int_0^p(1-t)^{-2}dt$, Tonelli's theorem gives the stated conditional expected-minimum-odds formula.

For necessity of the integrability criterion, take canonical interims $(x_i,y_i)$ and threshold measures $Q_i$. For any set $\mathcal I$ of $k+1$ agents, capacity implies
\[
 \sum_{i\in\mathcal I}\mathbf1_{\{\theta_i=1,p_i>t\}}(1-a_i)
 \ge\mathbf1_{\{\theta_i=1,p_i>t\ \forall i\in\mathcal I\}}.
\]
Expectations and monotonicity give
\[
 T_1(t)^{k+1}\le\sum_{i\in\mathcal I}\int_t^1p[1-x_i(p)]\,dF(p)
 \le T_1(t)\sum_{i\in\mathcal I}[1-x_i(t)].
\]
Divide by $T_1(t)>0$ and integrate against $(1-t)^{-2}dt$. By \eqref{appimpl:tailcharge} and the odds bounds,
\begin{equation}
 \frac{u_+}{u_-}I_k(F)
 \le\frac{u_+}{u_-}\sum_{i\in\mathcal I}\int v\,dQ_i\le k+1.
 \label{appimpl:necessary-ratio}
\end{equation}
Thus $I_k(F)<\infty$ is necessary.

For sufficiency of the criterion and part (ii), select suitable agents in decreasing report order, up to capacity, with fixed tie-breaking. Their suitable interims and partial odds moments are
\[
 X(p)=\Pr\{\Bin(n-1,T_1(p))\le k-1\},\qquad
 J_X(p)=\int_{(0,p]}v\,dX.
\]
The function $X$ is nondecreasing with $X(1-)=1$. Since
$\Pr\{\Bin(n-1,\eta)\ge k\}/\eta^k$ is positive and continuous on $(0,\mu]$ and tends to $\binom{n-1}{k}$ at zero, $1-X(t)$ is bounded above and below by positive multiples of $T_1(t)^k$. Tonelli gives
$J_X(1-)=\int_0^1[1-X(t)](1-t)^{-2}dt$, finite exactly when $I_k(F)$ is finite. The union bound further gives
\begin{equation}
 1-X(t)\le\binom{n-1}{k}T_1(t)^k,\qquad
 J_X(1-)\le\binom{n-1}{k}I_k(F).
 \label{appimpl:priority-odds-bound}
\end{equation}
Provisionally fill remaining places uniformly among unsuitable agents. Conditional on being unsuitable, each agent is selected with probability $(k-S^{\mathrm{FB}})/[n(1-\mu)]$, independent of its report: suitable reports affect identities but not the number selected. Retain an unsuitable agent of report $p$ with probability
\[
 \frac{n(1-\mu)u_+J_X(p)}{u_-[k-S^{\mathrm{FB}}]}.
\]
Bound~\eqref{eq:preservation-sufficient} makes this at most one. The resulting interims $(X,(u_+/u_-)J_X)$ have representation \eqref{eq:Qrepresentation} and satisfy the odds bound, so Lemma~\ref{lem:a-ic} gives incentives. Construction ensures pointwise capacity and selects $\min\{k,\sum_i\theta_i\}$ suitable agents. It therefore implements $S=S^{\mathrm{FB}}$. Since $I_k(F)>0$ and
$k-S^{\mathrm{FB}}=\E[(k-\Bin(n,\mu))_+]>0$, the sufficient bound holds for all sufficiently small positive ratios whenever $I_k(F)<\infty$.

If $I_k(F)=\infty$, necessity makes both private-belief ranges empty; set $\bar u_+^{\mathrm{FB}}=\bar u_+=0$. For the rest of this proof, suppose $I_k(F)<\infty$.

For the critical interval, fix $u_->0$ and restrict to jointly feasible threshold
measures by the preceding lemmas. Lemma~\ref{lem:a-compactness} gives
weak compactness at each positive benefit.

Feasibility is downward closed in $u_+$. To reduce the benefit from $u_+^{(2)}$ to $u_+^{(1)}$, retain suitable recommendations and thin unsuitable ones by $u_+^{(1)}/u_+^{(2)}$. This preserves capacity and keeps the same threshold measure with its required new unsuitable interim. When $I_k(F)<\infty$, the feasible-benefit set $\mathcal U$ is nonempty and bounded above by $u_-(k+1)/I_k(F)$. Let $u_{+,j}\uparrow\bar u_+:=\sup\mathcal U$ with feasible measures $Q_j$. Their odds moments are uniformly bounded, so a subsequence converges weakly to $Q$, with
\[
 \int v\,dQ\le\liminf_j\int v\,dQ_j\le\frac{u_-}{\bar u_+}.
\]
The bounded continuous coefficients in \eqref{appimpl:compactconstraints} permit passage to the limit as $u_{+,j}\to\bar u_+$. Thus the endpoint is feasible and $\mathcal U=(0,\bar u_+]$.

We first show that any jointly feasible threshold measure $Q$ with
$y_Q(1-)<1$ permits coordination at a strictly larger benefit, keeping $u_-$ fixed.
Apply the construction in Lemma~\ref{lem:b-endpoint}. In its notation, the
resulting $Q_\gamma$ has $\widehat y:=y_{Q_\gamma}(1-)<\ell_*<1$ and, at the
original payoff pair,
\[
 L_{Q_\gamma}(\ell)\le L(\ell)-d
 \quad(0\le\ell\le\ell_*),\qquad
 d:=\gamma(\Delta_*-E_t)>0.
\]
For $\lambda>1$, keep $Q_\gamma$ fixed and replace $u_+$ by
$\lambda u_+$. The suitable interim is unchanged and the unsuitable interim
becomes $\lambda y_{Q_\gamma}$. Choose $\lambda$ sufficiently close to one that
\[
 \lambda\widehat y<\ell_*,\qquad
 (\lambda-1)(1-\mu)\widehat y<d.
\]
The odds bound continues to hold because the new unsuitable endpoint is less
than one. On $[0,\ell_*]$, the one-Lipschitz property of the positive part
bounds the increase in the stop-loss transform by
\[
 (\lambda-1)\int_0^1(1-p)y_{Q_\gamma}(p)\,dF(p)
 \le(\lambda-1)(1-\mu)\widehat y<d.
\]
Above $\ell_*$, both old and new unsuitable positive parts vanish, and the
suitable terms are unchanged. All stop-loss inequalities therefore continue
to hold. Lemma~\ref{lem:a-capacity} implements the resulting interims at
$(\lambda u_+,u_-)$.

The set of benefits admitting preservation is nonempty by the decreasing-report construction above, is bounded above by $\bar u_+$,
and is downward closed by thinning unsuitable recommendations. Its positive
supremum $\bar u_+^{\mathrm{FB}}$ is attained. Indeed, take preserving measures
$Q_j$ at benefits increasing to this supremum. Their odds moments are uniformly
bounded, so a subsequence converges weakly. The compactness argument for
\eqref{appimpl:compactconstraints} passes feasibility to the limit, and the
bounded continuous function $T_1$ preserves
$n\int T_1\,dQ_j=S^{\mathrm{FB}}$. A preserving measure at the endpoint has
$y(1-)\le b<1$ by Lemma~\ref{lem:b-displacement}. The preceding construction
therefore gives a privately feasible benefit strictly larger than
$\bar u_+^{\mathrm{FB}}$, proving
$\bar u_+^{\mathrm{FB}}<\bar u_+$.

Finally, symmetrize any feasible mechanism at $u_+=\bar u_+$. If its canonical
interims had $y_i(1-)<1$ for even one agent, their symmetric average would have
endpoint strictly below one; the same construction would contradict maximality
of $\bar u_+$. Since every endpoint is at most one, all must equal one.

\end{proof}

\subsection{Proofs for observed beliefs}
\label{app:observed-beliefs}

The following lemma supplies the observed-belief comparison in Section~\ref{sec:costs} and the condition used in Proposition~\ref{prop:feasible-costs}(iii).

\begin{lemma}
\label{lem:observed-beliefs}
A feasible observed-belief policy exists if and only if \eqref{eq:observed-feasibility} holds. Whenever one exists, $(S,R)=(S^{\mathrm{FB}},0)$ is attainable. These conclusions do not require Assumption~\ref{ass:belief-support}.
\end{lemma}

\begin{proof}
Write $\varepsilon=u_+/u_-$. Take any feasible observed-belief mechanism and its specified equilibrium. Its conditional equilibrium action law given $(\mathbf p,\boldsymbol\theta)$ is a Borel kernel on $\{0,1\}^n$; write its suitable- and unsuitable-state interims as $x_i(p)$ and $y_i(p)$. Comparing equilibrium play with always using the resource gives
\begin{equation}
 \varepsilon p[1-x_i(p)]\le(1-p)[1-y_i(p)]
 \quad\text{for }F\text{-almost every }p.
 \label{appobs:original-obedience}
\end{equation}

The observed-belief advisor draws a preliminary action vector from this kernel. It retains the suitable agents in that vector, discards the unsuitable agents, and adds suitable nonusers until exactly $\min\{k,\sum_i\theta_i\}$ agents are selected. Use fixed index order for additions. If a preliminary vector has more than $k$ suitable users, retain its first $k$ instead. This last provision handles all exceptional profiles and draws: the resulting finite transformation is measurable and satisfies pointwise capacity at every belief--state profile. Only its final binary recommendations are sent; the preliminary draw and discarded recommendations are not disclosed. The advisor's observations of beliefs remain private.

The original capacity constraint holds almost surely, so the new suitable-state interim $\widehat x_i(p)$ is at least $x_i(p)$ for $F$-almost every $p$. Its unsuitable-state interim is zero. A use recommendation reveals suitability. For a nonuse recommendation, \eqref{appobs:original-obedience} gives
\[
 \varepsilon p[1-\widehat x_i(p)]
 \le(1-p)[1-y_i(p)]\le1-p
 \quad F\text{-almost everywhere}.
\]
Thus both recommendations are obeyed at the original payoffs. The rule selects exactly the direct-assignment number of suitable agents at every profile, and therefore attains $(S,R)=(S^{\mathrm{FB}},0)$.

If the original feasible mechanism instead elicits private beliefs, the observed-belief advisor can draw from its truthfully induced action kernel by substituting the observed beliefs for reports. The same construction and obedience argument apply, without retaining a reporting stage.

We next characterize existence. For necessity, take any feasible observed-belief outcome and its interims $x_i,y_i$. For each $t<1$, let
$N_t=\sum_i\mathbf1_{\{\theta_i=1,p_i>t\}}$.
Capacity implies that at least $(N_t-k)_+$ members of this group forgo use. Since $N_t$ has law $\Bin(n,T_1(t))$, \eqref{appobs:original-obedience} implies
\begin{align*}
 \varepsilon\big[nT_1(t)-H(T_1(t))\big]
 &\le\varepsilon\sum_i\int_t^1p[1-x_i(p)]\,dF(p)\\
 &\le\sum_i\int_t^1(1-p)[1-y_i(p)]\,dF(p)
 \le nT_0(t).
\end{align*}
This proves \eqref{eq:observed-feasibility}.

For sufficiency, suppose these inequalities hold, and define
\[
 x_{\min}(p)=\left[1-\frac{1-p}{\varepsilon p}\right]_+,
 \qquad p_0=\frac1{1+\varepsilon}.
\]
For $t\ge p_0$, the assumed inequality is exactly
\begin{equation}
 n\int_t^1p x_{\min}(p)\,dF(p)\le H(T_1(t)).
 \label{appobs:minimum-capacity}
\end{equation}
For $t<p_0$, the integral on the left equals its value at $p_0$, whereas the right side is at least its value at $p_0$. Thus \eqref{appobs:minimum-capacity} holds for every $t\in[0,1]$.

We apply the general coupling construction in the proof of Lemma~\ref{lem:a-capacity}; its capacity step uses only the relevant upper-set inequalities, rather than the threshold representation. In detail, put
$Z_{\min}(p,\theta)=\theta x_{\min}(p)$ under the law $\rho$,
and let $Z=\kappa(\omega)$ for an independent uniform $\omega$, with $\kappa$ as in that proof. For every $0\le\ell<1$, the set $\{Z_{\min}>\ell\}$ consists of suitable agents above a belief cutoff. Inequality~\eqref{appobs:minimum-capacity} and the rank identity \eqref{appimpl:rankidentity} therefore give
\[
 \E(Z_{\min}-\ell)_+
 \le\max_{0\le\eta\le1}\left\{\frac{H(\eta)}n-\ell\eta\right\}
 =\E(Z-\ell)_+
 \qquad(0\le\ell\le1).
\]
Strassen's theorem gives a Borel conditional score law $\Gamma(dz'\mid z)$ with marginal law $Z$ and conditional mean $\bar z(z)\ge z$. As in the lemma's proof, replace this law by $\delta_z$ on its exceptional null set so the mean inequality holds for every $z$. Given the observed profile, independently draw agent $i$'s score from $\Gamma(\cdot\mid Z_{\min}(p_i,\theta_i))$, select the $k$ highest scores with measurable tie-breaking, and retain a selected agent with probability $Z_{\min}(p_i,\theta_i)/\bar z(Z_{\min}(p_i,\theta_i))$, interpreted as zero when the denominator vanishes. Opponents' scores have the continuous marginal law $Z$, and an own score $z'$ is selected with probability $z'$. Averaging and retaining therefore give suitable interim $x_{\min}(p)$ and unsuitable interim zero at every type. The finite selection rule is a measurable kernel satisfying pointwise capacity.

A positive recommendation under this rule reveals suitability, and the definition of $x_{\min}$ gives
$\varepsilon p[1-x_{\min}(p)]\le1-p$.
Thus obedience holds. Completing its suitable selections as above implements $(S,R)=(S^{\mathrm{FB}},0)$. Neither this construction nor the necessity argument requires full support.
\end{proof}

\begin{proof}[Proof of Proposition~\ref{prop:feasible-costs}(iii)]
Write $\varepsilon=u_+/u_-$. Under full support, $T_1(t)>0$ for every $t<1$. Since $k<n$, the denominator $nT_1(t)-H(T_1(t))$ is strictly positive. Lemma~\ref{lem:observed-beliefs} therefore gives the critical value in \eqref{eq:observed-cutoff}, which is finite by taking $t=0$, and includes its endpoint whenever it is positive.

To characterize when the cutoff is positive, put
\[
 C:=\sup_{0\le t<1}\frac{T_1(t)^k}{1-t}.
\]
Condition~\eqref{eq:observed-tail} is $C<\infty$. Moreover, $t[1-F(t)]\le T_1(t)\le1-F(t)$ shows that it is equivalent to $1-F(t)=O((1-t)^{1/k})$ as $t\uparrow1$. Suppose an observed-belief policy is feasible at some $\varepsilon>0$. For $0<t<1$, write $a=T_1(t)$. The event that a fixed $k+1$ agents are suitable with beliefs above $t$ has probability $a^{k+1}$ and entails at least one excess agent. Hence \eqref{eq:observed-feasibility} gives
\[
 \varepsilon T_1(t)^{k+1}
 \le\varepsilon[nT_1(t)-H(T_1(t))]
 \le nT_0(t)
 \le n\frac{1-t}{t}T_1(t).
\]
Division bounds the ratio defining $C$ near one; $T_1(t)\le1$ bounds it on every interval $[0,t_0]$ with $t_0<1$. Thus $C<\infty$.

Conversely, suppose $C<\infty$ and select up to $k$ suitable agents in decreasing observed-belief order, with fixed tie-breaking. Its interims are
\[
 X(p)=\Pr\{\Bin(n-1,T_1(p))\le k-1\},\qquad y(p)=0.
\]
The union bound gives $1-X(p)\le\binom{n-1}{k}T_1(p)^k$. Consequently, any positive ratio satisfying $\varepsilon\binom{n-1}{k}C\le1$ obeys
\[
 \varepsilon p[1-X(p)]
 \le\varepsilon\binom{n-1}{k}C\,p(1-p)
 \le1-p.
\]
Use advice reveals suitability and nonuse is obedient, so the observed-belief feasible range is nonempty.

To prove the strict comparison of critical values, fix any privately feasible ratio $\varepsilon$ and a symmetric canonical threshold measure $Q$. Write
$M=\int v\,dQ\le1/\varepsilon$.
As in the capacity argument for \eqref{appimpl:nofiniteschedule}, $Q((t,1))>0$ for every $t<1$: otherwise suitable agents with beliefs above $t$ would all use the resource, violating capacity on a positive-probability event. In particular $M>0$. Choose $t_*\in(0,1)$ such that
$\int_{(t_*,1)}v\,dQ\le M/2$, and put
\[
 d=\min\left\{\frac M2,
       \int_{(t_*,1)}[v(\tau)-v(t_*)]\,Q(d\tau)\right\}>0.
\]
For $p\ge t_*$, $v(p)Q((p,1))\le M/2\le M-d$. For $p<t_*$, partitioning the odds moment at $p$ gives
\[
 M-v(p)Q((p,1))
 =\int_{[0,p]}v\,dQ
   +\int_{(p,1)}[v(\tau)-v(p)]\,Q(d\tau)
 \ge d.
\]
It follows that
\[
 v(p)[1-x_Q(p)]\le M-d<M\le1/\varepsilon
 \qquad(0<p<1).
\]
Choose $\varepsilon'>\varepsilon$ with $\varepsilon'(M-d)\le1$. Purging unsuitable selections and completing suitable selections from the original kernel gives suitable interims at least $x_Q$ almost everywhere. The resulting recommendations are therefore obedient at ratio $\varepsilon'$ as well. Every privately feasible ratio is strictly below the observed-belief critical ratio. When $\bar u_+>0$, applying this argument at the attained private endpoint from part (i) gives $\bar u_+<\bar u_+^{\mathrm{obs}}$.

\end{proof}

\subsection{Reporting requirements}
\label{app:reporting}

\begin{proof}[Proof of Proposition~\ref{prop:finite-reports}]
Suppose a feasible mechanism--equilibrium pair $(\Pi,\sigma)$ has a set
$\mathcal I$ of $k$ agents with finite report spaces. Write
$\alpha_i(r\mid p)=\sigma_i^{\mathrm r}(\{r\}\mid p)$.
Each $i\in\mathcal I$ has a report $r_i^*$ with
\[
 \int_t^1\alpha_i(r_i^*\mid p)\,dF(p)>0\qquad(t<1).
\]
Otherwise each report would disappear above some cutoff; finiteness would
contradict full support above the largest of these cutoffs.

Let $P_i^\theta$ be the law of $i$'s message when it reports $r_i^*$ and its
state is $\theta$, integrating opponents' primitives and equilibrium reports.
Choose Borel densities $f_i^\theta=dP_i^\theta/d(P_i^1+P_i^0)$.
At on-path information $(p_i,r_i^*,m_i)$, the gain from use is
\begin{equation}
 \frac{u_+p_i f_i^1(m_i)-u_-(1-p_i)f_i^0(m_i)}
 {p_i f_i^1(m_i)+(1-p_i)f_i^0(m_i)}.
 \label{appimpl:participationgain}
\end{equation}
The denominator is positive almost surely. Where $f_i^1>0$, this gain is
positive exactly when
\[
 p_i>\tau_i(m_i):=
 \frac{u_-f_i^0(m_i)}{u_+f_i^1(m_i)+u_-f_i^0(m_i)}<1;
\]
set $\tau_i=1$ where $f_i^1=0$.
Changing only continuation actions on a measurable set where the gain is
strictly positive would be profitable unless
\begin{equation}
 \sigma_i^{\mathrm a}(\{1\}\mid p_i,r_i^*,m_i)=1
 \quad\text{whenever }p_i>\tau_i(m_i),
 \label{appimpl:strict-continuation}
\end{equation}
outside a null set under the equilibrium law conditional on $r_i=r_i^*$.
This conditional law is well defined because that report has positive
probability. Thus ex ante equilibrium optimality supplies the required
on-path action optimality, including when reporting is randomized.

Define
\[
 \mathcal E=\bigcap_{i\in\mathcal I}\{\theta_i=1,r_i=r_i^*\},\qquad
 \eta_i=\int_0^1p\alpha_i(r_i^*\mid p)\,dF(p)>0,\qquad
 \rho_i^{\mathcal E}(dp)=\frac{p\alpha_i(r_i^*\mid p)}{\eta_i}\,dF(p).
\]
Then $\Pr_{\Pi,\sigma}(\mathcal E)=\prod_{i\in\mathcal I}\eta_i>0$, and
every $\rho_i^{\mathcal E}$ assigns positive probability to $(t,1)$ for $t<1$.
For each Borel message set $\mathcal B_i$, retaining opponents' realizations
in $\mathcal E$ gives
\[
 P_i^1(\mathcal B_i)\ge
 \left(\prod_{j\in\mathcal I\setminus\{i\}}\eta_j\right)
 \Pr_{\Pi,\sigma}\{m_i\in\mathcal B_i\mid\mathcal E\}.
\]
Thus $f_i^1(m_i)>0$ and $\tau_i(m_i)<1$ almost surely conditional on
$\mathcal E$.

Realize the behavioral kernels using independent private random variables.
Conditional on $\mathcal E$, the advisor's message kernel depends on these
agents only through their fixed reports and states. Let $\mathcal G$ be
the information generated by the entire message vector and by all
primitives, reports, and private randomization of agents outside
$\mathcal I$. Conditional on $\mathcal E$ and $\mathcal G$, the remaining
types have product law $\bigotimes_{i\in\mathcal I}\rho_i^{\mathcal E}$.
The null sets in \eqref{appimpl:strict-continuation} remain null after
conditioning on the positive-probability event $\mathcal E$. Fubini's
theorem and that implication therefore give, almost surely on $\mathcal E$,
\[
 \Pr_{\Pi,\sigma}\{a_i=1\text{ for all }i\in\mathcal I
                     \mid\mathcal E,\mathcal G\}
 \ge\prod_{i\in\mathcal I}\rho_i^{\mathcal E}((\tau_i(m_i),1))>0.
\]
Every factor is positive because $\tau_i(m_i)<1$ and
$\rho_i^{\mathcal E}$ has positive mass above every interior cutoff.
An outside agent's action is $\mathcal G$-measurable. If it used the resource
on a positive-probability subset of $\mathcal E$, the last display would
give positive probability that it and all $k$ agents in $\mathcal I$
used the resource, violating capacity.
Every outside agent must therefore forgo use almost surely on $\mathcal E$.
There is an outside agent $j$ because $n>k$.
The event $\mathcal E$ is independent of $(p_j,\theta_j)$, so its equilibrium
suitable-state action interim is at most $1-\Pr_{\Pi,\sigma}(\mathcal E)$
for $F$-almost every $p$.
Retaining the equilibrium report and then always using the resource is an admissible
deviation. If $x_j(p),y_j(p)$ are this agent's equilibrium action interims,
allowing this deviation on measurable sets of types gives, for $F$-almost every $p$,
\[
 u_+p[1-x_j(p)]\le u_-(1-p)[1-y_j(p)]\le u_-(1-p).
\]
Thus $u_+p\Pr_{\Pi,\sigma}(\mathcal E)\le u_-(1-p)$ almost everywhere,
contradicting full support near one. Thus at most $k-1$ agents can have finite
report spaces.

\end{proof}

\begin{singlespace}
\small
\interlinepenalty=10000
\bibliographystyle{plainnat}
\bibliography{references}
\end{singlespace}

\end{document}